\documentclass[journal]{IEEEtran}
\usepackage{blindtext}
\usepackage{graphicx}
\usepackage{booktabs} 
\usepackage{amsbsy}
\usepackage[ruled]{algorithm2e} 
\usepackage{multirow}
\usepackage{amssymb}
\usepackage{amsmath}
\usepackage{tikz}
\usepackage[utf8]{inputenc}
\usepackage[english]{babel}
\usepackage{color}
\usepackage{soul}
\usepackage{fancyhdr}
 
\newtheorem{theorem}{Theorem}[section]
\newtheorem{corollary}{Corollary}[theorem]

\newtheorem{definition}{Definition}[section]
\usetikzlibrary{automata, positioning}

\fancypagestyle{firstpage}{
  \fancyhf{}
  \fancyhead[L]{\footnotesize{THIS PAPER HAS BEEN ACCEPTED FOR PUBLICATION IN IEEE TRANSACTIONS ON  RELIABILITY, OCT. 2018}}
  \fancyhead[R]{\footnotesize{\thepage}}
}

\ifCLASSINFOpdf
\else
\fi

\begin{document}
\title{A Reliability Model for Dependent and Distributed MDS Disk Array Units}

\author{Suayb~S.~Arslan,~\IEEEmembership{Member,~IEEE}
\thanks{Suayb S. Arslan  is with the Department
of Computer Engineering, MEF University, Maslak,
Istanbul, Turkey, e-mail:arslans@mef.edu.tr (see http://www.suaybarslan.com/contact.html).}
\thanks{This work is supported by both Quantum Corporation, San Jose, CA, USA and The Scientific and Technological Research Council of Turkey under grant number 2232-115C111.}}


\maketitle
\thispagestyle{firstpage}

\begin{abstract}
Archiving and systematic backup of large digital data generates a quick demand for multi-peta byte scale storage systems. As drive capacities continue to grow beyond the
few terabytes range to address the demands of today's cloud, the likelihood of having multiple/simultaneous disk failures become a reality. Among the main factors causing catastrophic system failures, correlated disk failures and the network bandwidth are reported to be the two common source of performance degradation. The emerging trend is to use efficient/sophisticated erasure codes (EC) equipped with multiple parities and efficient repairs in order to meet the reliability/bandwidth requirements. It is known
that mean time to failure and repair  rates reported by the disk manufacturers cannot capture life cycle patterns of distributed storage systems.  In this study, we develop failure models based on generalized Markov chains that can accurately capture correlated performance degradations with multi-parity protection schemes based on modern Maximum Distance Separable (MDS) EC. Furthermore, we use the proposed model in a distributed storage scenario to quantify two example use cases: Primarily, the common sense that adding more parity disks are only meaningful if we have a decent decorrelation between the failure domains of storage systems and the reliability of generic multiple single-dimensional EC protected storage systems.
\end{abstract}


\IEEEpeerreviewmaketitle

\section{Introduction} \label{Intro}

Increased gap between the capacity  and the input/output data access rates of commercial disks, coupled with the increased appeal for thousands of small component commodity storage units, have lead to the development of disk arrays. However, incorporating such a large volume of disks into the array leads to increased and correlated failure rates, even in some cases worse than that of a single disk \cite{Pinheiro}. Large number of installations of such disk arrays result in an overall decreased reliability. For example, it is well known that the extensions of the Redundant Array of Inexpensive Disks (RAID) \cite{RAID} systems are envisioned to tolerate situations in which two or more disk failures happen due to increased failure rates \cite{Cleversafe}. In case of reconstruction or the so called \emph{repair} process of the failed component disks, excessive read requests for data regeneration might have to be serviced due to the increased capacities and therefore, the recovery process becomes susceptible to incumbent read errors as well as the network failures. This is another reason that the traditional parity-based RAID (e.g. RAID 5 and RAID 6 \cite{Corbett}) systems fail to meet today's reliability requirements for digital data storage.

Rising trend for storing large volumes of data led to improvements on basic RAID (e.g. efficient implementations of RAID 6), forcing manufacturers to add extra parity disks to RAID 5 setting in order to boost the reliability performance of disk arrays. All versions of RAID are typically implemented in hardware and are based on erasure codes with the optimal capacity-recovery property, known as \emph{maximum distance separable} (MDS) constraint. Especially, when the stored data is of small volume and the scale of the storage system is moderate, RAID techniques were found to be excellent options with enough user data protection. While the scale of storage systems expand and the requirements of different applications change over time, reliability and scalability of RAID systems became questionable \cite{Schroeder} which led to some of the research efforts to search for techniques at the disk array level to improve RAID's reliability \cite{Dholakia}.

When the component disks happen to be in the same geographical location, or mounted in the same network storage node, correlated failures become the performance bottleneck. For example, failures within a batch of disks are observed to be strongly correlated \cite{Schroeder}. Disks that belong to the same manufacturer usually go through the same manufacturing process and made of the same type of magnetic and electronic materials. Their similarity does not decrease dramatically even if the manufacturers are different, because the core materials used in the production phase are similar, if not the same. Furthermore, disks that end up in the same box or a network storage node are subject to the same type of environmental conditions. Such environmental conditions affect the overall disk array reliability almost the same way under normal circumstances. Plus, such disks share the same support hardware. Whenever a catastrophic error occurs in the hardware \cite{GIBSON}, it can easily cause multiple and simultaneous disk failures.

\subsection{Related Work}
With the raise of modern erasure codes that allow network-efficient repairs \cite{Dimakis} and minimize the data read times while servicing user data requests degraded reads or data regeneration requests, the time it takes to maintain the system operability, repair the data and the hardware, balance the system with necessary data transmissions have completely transformed the old reliability problem into very hard one to predict. However, for a reliable and optimized system design such predictions are very crucial and necessary. 

In \cite{Molhotara}, the first Markov chain reliability analysis of disk arrays is performed. Following this study, slight generalizations have been made to the basic model \cite{Burkhard}, \cite{Blum}. For instance in \cite{Elerath}, a kind of enhanced reliability modelling is proposed. Later, more realistic failure phenomena are introduced to the model to address accuracy problem such as latent error and bit rod cases \cite{Iliadis}. In \cite{Hafner}, it is shown that with few more generalizations, the basic model can also be used for non-MDS disk arrays. Although \cite{Greenan} disputes with conventional metrics and proposes a new one based on the average data loss, they fail to provide a comprehensive model and closed form expressions that capture the correlated nature of failures. In some previous works such as \cite{Nath}, subtleties regarding correlated failures is considered. However no specific reliability modelling is proposed. On the other hand, studies like \cite{YanLi} considers only little work on quantifying correlated failure problem and providing a framework to minimize its impact. Rather, a model is proposed to calculate the survivability of data objects stored on heterogeneous storage systems.  

To the best of our knowledge, all of these previous reliability models proposed for disk arrays cannot accurately capture some real time phenomena such as common failure dependencies or accurately predicting the lifespan of storage systems protected by modern erasure codes that use the network and computation resources effectively \cite{Rashmi}. Inspired by this observation, we propose a generalized Markov model that can be used to analyze disk failures under such dependent factors. For instance using the proposed model, we are able to validate and accurately quantify an experimental observation that adding more parity only has significant effect on reliability if we have independent disk failure rates which was conventionally identified and compensated by declustering methods \cite{Gang}. Particularly, we shall argue that if disk failure rates depend on the number of previously failed disks, then this argument is no longer true. For a given exponential failure rate growth model, we show that we can exactly quantify the maximum number of parity disks (e.g. 4 or 5) beyond which adding more parity disks has practically no effect on the overall reliability of the system and therefore can be considered to be a waste of resources. On the other hand, the proposed model will be shown to be useful for estimating the reliability of disk arrays which are protected by modern and sophisticated erasure coding schemes such as pyramid codes \cite{Huang}. The ability to incorporate new metrics such as repair bandwidth, average read overhead, etc. in to the model is deemed to be very important for the reliability estimation of the next generation distributed and networked storage systems.

Most of the MDS erasure correcting codes are applied to a series of disks, constituting MDS disk array schemes. In general, such a class of MDS-based protection schemes are considered as $t$ dimensional ($t$-D) MDS disk protection schemes to create a more robust system against disk failures. For example, a special case of such class of MDS-based protection schemes is considered in \cite{GIBSON}. We mostly focus on one dimensional MDS protection schemes in this study, however it can be shown that the model can be used to derive closed form expressions for multi dimensional MDS disk arrays \cite{Arslan} and code structures \cite{Arslan1}. Furthermore, although it is out of scope of this paper, we can show that the general model proposed in this study can be used to predict reliability for non-MDS codes  \cite{Lee, Arslan2}, array BP-XOR codes\cite{Arslan3, Wang2013} as well as repair--efficient codes \cite{Li2015},\cite{Park2018} under novel opportunistic repair mechanisms \cite{Aggarwal2014}.

More recently, information dispersal has gained attraction due to its reliable operation compared to conventional disk arrays \cite{Cleversafe}. For a given distributed information scenario, we also consider in this study a network storage system consisting of few nodes and two commonly employed disk allocation strategies (\textit{horizontal} and \textit{vertical}) for stripped $1$-D MDS-protected disk arrays. We analytically evaluate their reliability based on the proposed general Markov model and argue that given a correlated disk failure growth model, information dispersal achieves more reliable data protection. Alternatively, we also argue that for a given reliability target, it leads to less use of redundant disks, particularly using vertical allocation. The latter ultimately means considerable amount of savings in terms of resources without compromising the target reliability. Finally,  we remark that allocation is a critical part of the information dispersal paradigm and we only considered two straightforward methods in this study. We anticipate to consider the reliability analysis of more advanced data allocation strategies as our future work.

\subsection{Organization}

The rest of the paper is organized as follows. In section \ref{MDSEPG}, we provide brief information about MDS disk array systems and the concept of Error Protection Groups (EPGs). In Section \ref{GMFM}, we introduce the general failure model based on Markov chains along with a low complexity scheme to calculate the mean time to data loss. We also give failure growth rates of interest as functions of the number of  operational disks of protection groups. Moreover, we show the extension of the general model to cover advanced features such as hard errors, initial defective disks and average read overhead. In Section \ref{Distributed}, we consider an information dispersal scenario using horizontal and vertical allocations of disks in a distributed setting using the proposed failure model. Some numerical results are provided at the end of both Section \ref{GMFM} and Section \ref{Distributed}. Finally, Section \ref{CNCLSN} concludes the paper.  Some of the short proofs are included within the text whereas more sophisticated ones are moved to the end of the paper in appendices \ref{app1}-\ref{app3} to make the paper flow smoothly.

\section{Storage Arrays Based on MDS Codes and Error Protection Groups} \label{MDSEPG}

Recent developments in hard and solid state disk array industry and well known experimental survey data \cite{Schroeder} confirmed complex, dependent and non-uniform failure rates across the constituent storage units for large scale data storage applications. In particular, disk replacement rates show significant correlation between constituent disks and are no where near manufactures' reported disk failure rates. 

For storage space efficiency, disk array systems typically use Reed Solomon-based \cite{RS} MDS erasure codes with efficient implementations of encoding and decoding processes \cite{Plank}, \cite{Plank2}.  All component disks over which the parity information is computed and at which the computed redundancy is stored is called an EPG. A MDS $(n,m)$ erasure correcting code is applied to the $m$ data disks to generate $p = n-m$ parity disks in order to make up a $n$-disk EPG. Since the code is MDS, it can recover up to $p$ failed disks in an EPG. Disk stripping is used to allocate data and parity information units across and along the disk arrays in order not to have dedicated parity or data disks \cite{Salem} for better I/O performance. Fig. \ref{fig:RAID}.a shows an example EPG i.e., 1-D MDS disk array, consisting of $m$ data and $p$ parity disks made up of multiple storage units such as sectors in hard disk systems. Note that these protection schemes are not limited to single dimension. An example for 2-D MDS disk arrays is shown in Fig. \ref{fig:RAID}.b where in addition to horizontal MDS encoding, there is also a vertical MDS encoding that provides extra robustness against disk failures. Each disk in the array is part of two different EPGs and the data can efficiently be recovered by a collaborative  decoding of EPGs using iterative algorithms. In general, any subset of data blocks can be used to form EPGs \cite{Huang} which may provide additional advantages for different use cases.

\begin{figure}[t!]
\centering
\includegraphics[height = 40mm, width=\columnwidth]{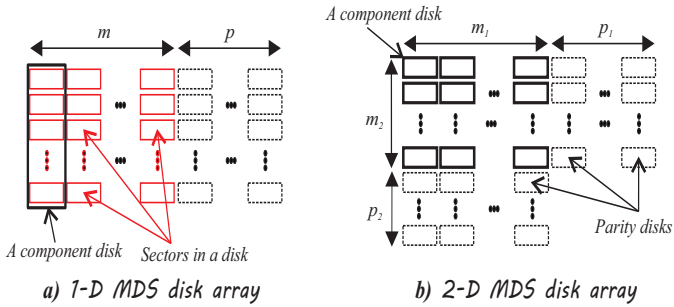}
\caption{a) An EPG consists of $m$ data and $p$ parity disk units generated by using a $(m+p,m)$ MDS code. This is also known as 1-D MDS disk array. b) A 2-D MDS disk array is shown in which data disks are both horizontally and vertically encoded using $(m_1+p_1,m_1)$ and $(m_2+p_2,m_2)$ MDS codes, respectively. }\label{fig:RAID}
\end{figure}

\begin{figure*}[t!]
\centering
\begin{tikzpicture}[node distance=1.1cm]
    \footnotesize
    \node[state, minimum width=1cm, line width=0.4mm]             (n) {$ n $};
    \node[state, minimum width=1cm, line width=0.4mm, right=of n] (n1) {$n-1$};
    \node[state, minimum width=1cm, line width=0.4mm, right=of n1, draw=none] (n2) {$\dots$};
    \node[state, minimum width=1cm, line width=0.4mm, right=of n2] (j) {$j$};
    \node[state, minimum width=1cm, line width=0.4mm, line width=0.4mm, right=of j, draw=none] (n3) {$\dots$};
    \node[state, minimum width=1cm, line width=0.4mm, right=of n3] (m1) {$m + 1$};
    \node[state, minimum width=1cm, line width=0.4mm, right=of m1] (m) {$m$};
    \node[state, minimum width=1cm, line width=0.4mm, below=of m] (f) {$F$};

    \draw[every loop]
        (n) edge[bend left, auto=left] node {$n\lambda_0$} (n1)
        (n1) edge[bend left, auto=right]  node {$\mu_0$} (n)
        (n1) edge[bend left, auto=left] node {} (n2)
        (n2) edge[bend left, auto=left] node {$(j+1)\lambda_{n-j-1}$} (j)
        (j) edge[bend left, auto=left] node {$j\lambda_{n-j}$} (n3)
        (n3) edge[bend left, auto=left] node {} (m1)
        (m1) edge[bend left, auto=left] node {$(m+1)\lambda_{n-m-1}$} (m)
        (m) edge[bend left=30, auto=right]  node {$(n-m)\mu_{n-m-1}$} (n)
			(j) edge[bend left=30, auto=right]  node {$(n-j)\mu_{n-j-1}$} (n)
        (m) edge[bend left, auto=left] node {$m\lambda_{n-m}$} (f)
        (n) edge[bend right=22, auto=left, dashed] node {$\gamma_0$} (f)
        (n1) edge[bend right=7, auto=left, dashed] node {$\gamma_1$} (f)
        (j) edge[bend right=4, auto=left, dashed] node {$\gamma_{n-j}$} (f)
        (m1) edge[bend right=14, auto=left, dashed] node {$\gamma_{n-m-1}$} (f);
\end{tikzpicture}
\caption{A generalized Markov failure model in which the labels on each state designate the number of operational number of disks in an EPG.  In general, we have the relationship $\lambda_0 < \lambda_1 < \dots < \lambda_{n-m}$ to describe the increasing failure
rates as we have more and more disks fail within the same EPG. \emph{F}: Failure state.}\label{fig:MARKOV}
\end{figure*}
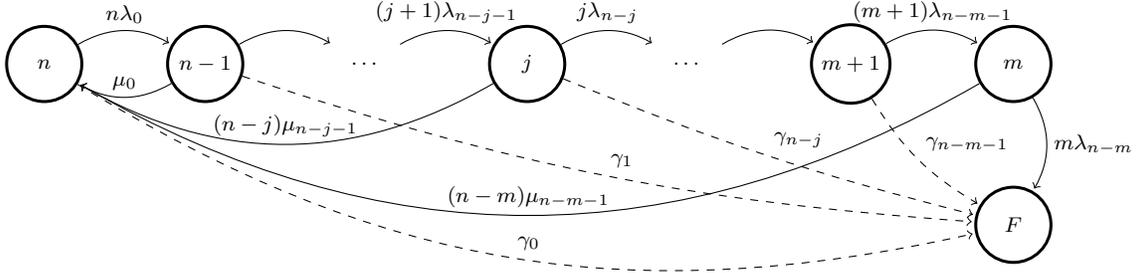

In order to characterize the data durability  in terms of Mean Time to Data Loss (MTTDL) metric for general subsets of EPGs using MDS codes, we need a generalized Markov model. This model needs to capture the dependent failure as well as concurrent repair rates. Plus it should involve error rates that can model total system crashes no matter how many extra parities are available. We will show that such a generalized model will be sufficient to capture different realistic scenarios for MDS protected storage arrays.

\section{A Generalized Markov Failure Model} \label{GMFM}

We use  MTTDL metric to quantify the average reliability of a given protected array of storage devices. Although there are studies arguing that  MTTDL is a deficient tool for the absolute measurements \cite{Greenan} and a system designer may be interested in the probability of failure for the first few years instead of the mean time to failure, MTTDL is still one of the most widely used intuitive reliability estimation metric that helps system designers to make accurate decisions.  Moreover, this notorious reliability metric, based on exponential failure and repair times, has been shown to be insensitive to the actual distribution of failure/repair times as long as the constituent storage devices have much larger mean time between failures than the mean repair time \cite{Venkatesan} and operate independent of each other. Thus, using MTTDL and known distributions we can simply generate an answer to the probability of failure for the first few years pretty accurately. Plus, closed form expressions for MTTDL shall be shown to be possible in this study for a general case and such analytical expressions usually help our intuition for modeling error--tolerant data storage systems.

The reliability characteristics of many storage devices follow what is usually known as "bathtub curve" \cite{Yang}. This curve is a composite of decreasing, constant and increasing failure rates at different times of the device lifetime. When disks are put into the service, the internal defective components fail quite rapidly. This leads to increased failure rates and the time where such failures take place is described as ``infant mortality period". When the disk enters into steady state in which only random errors dominate, the failure rates show steadiness. Thus, disks are in ``useful life period" and show constant failure rates. However, disks exhibit constant failure rates in their useful life period only if they work individually \cite{Modarres}. A simple failure model is given in \cite{Burkhard} which is considered to be a sufficient model for reliability estimations of EPGs containing $n$ disk units using a MDS code. As the disks age and wear due to different types of stresses and physical damages, the ``wear out period" kicks in and the failure rates start increasing again.

In an EPG, there are more than one disk component which work concurrently. In addition, a subset of disks may share the same hardware backend, once failed it will disallow access to all of the subset of disks. As argued before for such concurrent operations, a correct failure model must accommodate the dependent failures. Thus, a failure in an EPG will have an effect on the failure rates of the remaining disk components. In order to describe such a dependency with a simple model, we begin with a classical independent failure assumption and let the failure rate to vary based on the number of failed disks in the disk array.

We propose to use the generalized Markov model shown in Fig. \ref{fig:MARKOV}.  We assume a disk failure rate of $\lambda_0$ and a repair rate of $\mu_0$ at the beginning of operation. We also allowed transitions from any state $j, n \geq j \geq m+1$ to the failure state $F$. The rate of these transitions is called error rates and are quantified by $\gamma_i$. We can use error rates to model device dependent hard failures, multiple MDS array systems, non-MDS protected disk arrays or incorporate more realistic features. Thus, the proposed failure model is completely determined by  the parameters  $n$, $m$, the set of failure  rates $\boldsymbol{\lambda} = \{\lambda_0,\dots,\lambda_{n-m}\}$, repair rates $\boldsymbol{\mu} = \{\mu_0,\dots,\mu_{n-m-1}\}$ and error rates $\boldsymbol{\gamma} = \{\gamma_0,\dots,\gamma_{n-m-1}\}$.

In our model, the labels on each state designate the number of operational and accessible disks in an EPG. As can be observed from Fig. \ref{fig:MARKOV}, if more than $p=n-m$ disk failures happen at anytime, then the system goes into failure state ($F$). In the proposed general model, the disks are assumed to be repaired individually and the repair process produces all repaired disks at once i.e., concurrent maintenance. Particularly, since we are interested in 1-D MDS arrays, it is unnecessary to introduce $\gamma_i$s in which case it becomes possible to quantify MTTDL in a closed form. For better tractability, we will focus on the reduced model with $\gamma_i = 0$ in the next section and then extend results to comprise the more general model given in Fig. \ref{fig:MARKOV}.

\begin{figure*}[htp!]
    \centering
\begin{eqnarray}
L_{P_n}(s)s + \lambda_0 nL_{P_n}(s)  - \mu_0L_{P_{n-1}}(s) - 2\mu_1 L_{P_{n-2}}(s)- \dots -  p\mu_{n-m-1} L_{P_{m}}(s) = P_n(0) &=& 1  \nonumber \\
-\lambda_0nL_{P_n}(s) + L_{P_{n-1}}(s)s + \lambda_1 (n-1)L_{P_{n-1}}(s) +\mu_0L_{P_{n-1}}(s)  = P_{n-1}(0) &=& 0  \nonumber \\
-\lambda_1(n-1)L_{P_{n-1}}(s) + L_{P_{n-2}}(s)s + \lambda_2 (n-2)L_{P_{n-2}}(s) +2\mu_1L_{P_{n-2}}(s)  = P_{n-2}(0) &=& 0 \nonumber \\
 \vdots \ \ \ \ \ \ \ \vdots \ \ \ \ \ \ \ \vdots \ \ \ \ \ \ \ \vdots \ \ \ \ \ \ \  \vdots \ \ \ \ \ \ \ &=& 0 \nonumber \\
 -\lambda_{n-m-1}(m+1) L_{P_{m+1}}(s)  + L_{P_{m}}(s)s + \lambda_{n-m}mL_{P_{m}}(s) + p\mu_{n-m-1}L_{P_{m}}(s) = P_{m}(0) &=& 0 \nonumber  \\
 -\lambda_{n-m}mL_{P_{m}}(s) +  L_{P_{F}}(s)s  = P_{F}(0) &=& 0 \nonumber
\end{eqnarray}
\end{figure*}


\subsection{Mean time to Data Loss Performance}

The MTTDL is a measure used to quantify the average time before the storage system goes into the failure state $F$. Let us associate with each  state $P_i(t)$, the probability of being in state $i$ at time $t$. Using a similar notation with \cite{Burkhard}, the reliability function $R(t)$ is the probability of being in one of the states $n,\dots,m$ at time $t$, and is given by
\begin{eqnarray}
R(t) = P_{n}(t) + P_{n-1}(t) + \dots + P_{m+1}(t) + P_{m}(t) \label{Rt}
\end{eqnarray}

Now suppose that the disk lifetime random variable $t$ has the probability density function $f(t)$ and the reliability function $R(t) \triangleq \int_{t}^{\infty} f(x) dx$. The MTTDL using $p$ parity disks (denoted by $MTTDL_p$) is defined as \cite{Modarres}
\begin{eqnarray}
MTTDL_p &\triangleq& \mathbb{E}[t] = \int_0^\infty tf(t)dt = \int_0^\infty t\left(\frac{-dR(t)}{dt}\right)dt \nonumber \\
&=& -tR(t)|_0^\infty + \int_0^\infty R(t)dt = \int_0^\infty R(t)dt
\end{eqnarray}

Laplace transform of the reliability function in Eqn. (\ref{Rt}) i.e., $L_R(s)$ is instrumental for evaluating the integral,
\begin{eqnarray}
L_R(s) \triangleq \int_0^\infty R(t)e^{-st}dt \\ \Longrightarrow L_R(0) =  MTTDL_p = \sum_{j=m}^{m+p}  L_{P_j}(0)
\end{eqnarray}

We can write the following state  transition equations based on Fig. \ref{fig:MARKOV} with $\gamma_i = 0$,
\small
\begin{eqnarray}
&& \frac{dP_n(t)}{dt} + \lambda_0 nP_n(t)  - \mu_0 P_{n-1}(t) - 2\mu_1 P_{n-2}(t) - \dots \nonumber \\ && \hspace{5.2cm} -  p\mu_{n-m+1} P_{m}(t) = 0  \nonumber \\
&& -\lambda_0nP_n(t) + \frac{dP_{n-1}(t)}{dt} + (\lambda_1 (n-1) +\mu_0 )P_{n-1}(t)  = 0  \nonumber \\
&& -\lambda_1(n-1)P_{n-1}(t) + \frac{dP_{n-2}(t)}{dt} + (\lambda_2 (n-2) +2\mu_1)P_{n-2}(t)  = 0  \nonumber \\
&& \vdots \ \ \ \ \ \ \ \vdots \ \ \ \ \ \ \ \vdots \ \ \ \ \ \ \ \vdots \ \ \ \ \ \ \  \vdots \ \ \ \ \ \ \ = 0 \nonumber \\
&& -\lambda_{n-m-1}(m+1) P_{m+1}(t) +  \frac{dP_{m}(t)}{dt} + \dots \nonumber \\ && \hspace{3.65cm} (\lambda_{n-m}m + p\mu_{n-m-1})P_{m}(t) = 0  \nonumber \\
&& -\lambda_{n-m}mP_{m}(t) +  \frac{dP_{F}(t)}{dt}  = 0 \nonumber
\end{eqnarray}
\normalsize
with initial conditions  $P_n(0) = 1, P_j(0) = 0$ for $j = n-1, n-2, \dots, m+1, m, \emph{F}$ since all disks are assumed to be operational at the beginning. Taking the Laplace transform of each equation will yield linear equations in the transform domain.  Finally, we replace the coefficients of each linear equation to form a $p+2 \times p+2$ matrix $\mathbf{A}(s)$, as shown in the next page with $\gamma_i=0$.

Thus, using linear algebra we have the following equation to solve:
\begin{eqnarray}
\mathbf{A}(s) \mathbf{P}(s) = \mathbf{N}^{(0)}  \label{Eqn2solve}
\end{eqnarray}
where $\mathbf{P}(s) = [L_{P_n}(s)  \ \ \dots \ \  L_{P_{m+1}}(s) \ \ L_{P_m}(s) \ \ L_{P_F}(s)]^T$ and $
\mathbf{N}^{(l)} = [0 \ \ 0 \ \ \dots 0 \ \ 1 \ \ 0 \dots \ \ 0 \ \ 0]^T$ where $(l+1)$th array entry is unity. Once we solve for $\mathbf{P}(s)$, it is straightforward to compute $MTTDL_p = \sum_{j=m}^{m+p}  L_{P_j}(0)$.  Let us consider an example with $p=1$. By solving the Eqn. (\ref{Eqn2solve}), we obtain
\begin{eqnarray}
L_{P_{m+1}}(s) &=& (\lambda_1m + s + \mu_0)/\phi_1(s) \\
L_{P_{m}}(s) &=& (\lambda_0(m + 1))/\phi_1(s) \\
L_{P_{F}}(s) &=& (\lambda_0\lambda_1m(m+1))/s\phi_1(s)
\end{eqnarray}
where $\phi_1(s) = s(s+ \mu_0  + \lambda_0(m+1) + \lambda_1m ) + \lambda_0\lambda_1m(m+1)$
and the mean time to data loss is  therefore given by
\begin{eqnarray}
MTTDL_1 = \sum_{j=m}^{m+1}  L_{P_j}(0) = \frac{\lambda_0(m + 1) + \lambda_1m + \mu_0 }{\lambda_0\lambda_1m(m+1)} \label{MTTDL1}
\end{eqnarray}

Similarly, we can compute for $p=2$ as follows. 
\begin{eqnarray}
MTTDL_2 &=& \sum_{j=m}^{m+2}  L_{P_j}(0) \nonumber  \\
&=& \frac{(2\mu_1 + \lambda_2 m)(\lambda_0(m + 2) + \lambda_1(m+1) + \mu_0)}{\lambda_0\lambda_1\lambda_2m(m+1)(m+2)} \nonumber \\ && \ + \frac{1}{\lambda_2m} \label{MTTDL2}
\end{eqnarray}

As can be seen, expressions are getting more complex as we increase the number of parity disks/blocks, $p$ even in the absence of error rates $\gamma_i \in \boldsymbol{\gamma}$.

\begin{figure*} 
\normalsize
 \[\mathbf{A}(s) =  \left[ \begin{array}{cccccc}
s + \lambda_0 n + \gamma_0   & -\mu_0 & -2\mu_1 & \dots & -p\mu_{n-m-1} & 0 \\
-\lambda_0n & s + \lambda_1 (n-1) \atop +\mu_0 + \gamma_1 & 0 & \dots & 0 & 0 \\
0 & -\lambda_1(n-1) & s + \lambda_2 (n-2) \atop +2\mu_1 + \gamma_2 & \dots & 0 & 0 \\
0 & 0  & - \lambda_2 (n-2) & \dots & 0 & 0 \\
\vdots & \vdots &  & \ddots & 0 & 0 \\
0 & 0 & \dots &  -\lambda_{n-m-1}(m+1) & s + \lambda_{n-m}m \atop +p\mu_{n-m-1} & 0 \\
-\gamma_0 & -\gamma_1 & \dots & -\gamma_{n-m-1} & -\lambda_{n-m}m & s \end{array} \right].\]
\normalsize
\end{figure*}

\subsection{Efficient Computation of MTTDL}

As can be seen for large $n$ and $p$, it becomes harder to solve  Eqn. (\ref{Eqn2solve}) and numerically unstable to find $\mathbf{A}^{-1}(s)$. Particularly storage systems that use fountain--like codes to generate boundless number of parities \cite{Luby} and large number of network nodes for data distribution shall benefit from efficient and generalized formula for $MTTDL_p$. For a given EPG of size $n$ disks, we present below a straightforward method to efficiently compute $MTTDL_p$ for any $p \in \{1,2,\dots,n-1\}$.

For $x=0,1,\dots,p-1$, we begin by defining the following array with entries,
\small
\begin{eqnarray}
\boldsymbol{\Lambda}_x^{(p)} &\triangleq& \left[ \begin{array}{c}
\lambda_0 (m+p)  \\
\lambda_1 (m+p-1)   \\
\lambda_2 (m+p-2)  \\
\vdots   \\
\lambda_{p-1} (m+1)
\end{array} \right] \nonumber \\ && +
\mathbf{V}_x^{(p)}  \left( \left[ \begin{array}{c}
0  \\
\mu_0   \\
2\mu_1  \\
\vdots   \\
(p-1)\mu_{p-2}
\end{array} \right]  +
\left[ \begin{array}{c}
\gamma_0  \\
\gamma_1   \\
\gamma_2  \\
\vdots   \\
\gamma_{p-1}
\end{array} \right] \right)
\end{eqnarray}
where
\begin{eqnarray}
\mathbf{V}_x^{(p)}   = \begin{bmatrix}
\mathbf{0}_{x} & \mathbf{0}_{x \times p-x}   \\
\mathbf{0}_{p-x \times x} & \mathbf{I}_{p-x}
\end{bmatrix}
\end{eqnarray}
\normalsize
and  $\mathbf{I}_{p-x}$ and $\mathbf{0}_{x}$ represent identity  and all-zero matrices, respectively.

\vspace{1mm}
\begin{theorem} \label{thm31}
Let us assume $\Lambda_x^{p}(j)$ denote the $(j+1)$-th entry of the array $\boldsymbol{\Lambda}_x^{(p)} $. If we let $\gamma_i = 0$ then, we have the following transform domain expressions evaluated at $s=0$,
\begin{eqnarray}
L_{P_{m+p-x}}(0) &=& \frac{(p\mu_{p-1} + \lambda_pm)}{\phi_p(0)}\prod_{\substack{j=0 \\ j \not=x}}^{p-1}\Lambda_x^{p}(j)   \\ L_{P_{m}}(0) &=& \frac{1}{\phi_p(0)}\prod_{i=0}^{p-1} \lambda_i (m + 1 + i) \label{thm1}
\end{eqnarray}
where the denominator is given by
\begin{eqnarray}
\phi_p(0) = \prod_{i=0}^{p}\lambda_i (m + i) \label{Phip}
\end{eqnarray}
\end{theorem}

\begin{proof}
The proof is provided in Appendix A.
\end{proof}

From Eqns. (\ref{thm1}) and (\ref{Phip}), we can deduce $L_{P_{m}}(0) = 1/\lambda_pm$. Therefore, using Theorem \ref{thm31}  we find a closed form expression for $MTTDL_p$ as follows,
\begin{eqnarray}
MTTDL_p &=& 
\frac{(p\mu_{p-1} + \lambda_pm)}{\phi_p(0)} \sum_{x=0}^{p-1} \prod_{\substack{j=0 \\ j \not=x}}^{p-1}\Lambda_x^{p}(j) + \frac{1}{\lambda_pm} \nonumber
\end{eqnarray}

One can check the accuracy of the general form by setting $p=1$ and $p=2$ and comparing the results with Eqns (\ref{MTTDL1}) and (\ref{MTTDL2}). In addition, assuming a fixed failure and repair rates, i.e., $\lambda_i = \lambda$ and $\mu_{i-1}=\mu$ for all $i$, we can easily realize that these expressions are the generalized versions of the results found in \cite{Burkhard}.

Alternatively,  $MTTDL_p$ can be computed using numerical tools by recognizing that
\begin{eqnarray}
\mathbf{A}_{p+1}(s) \mathbf{P}_{p+1}(s) = \mathbf{N}^{(0)}_{p+1}
\end{eqnarray}
where the subscript $p+1$ denotes the upper left square submatrix if it is a matrix and the first $p+1$ entries if it is an array. Hence, $[L_{P_n}(0) \ \ L_{P_{n-1}}(0) \ \ \dots \ \  L_{P_{m-1}}(0) \ \ L_{P_m}(0)]  = \lim_{s \rightarrow 0}\mathbf{A}_{p+1}^{-1}(s)\mathbf{N}^{(0)}_{p+1}$. Therefore, we have
\begin{eqnarray}
MTTDL_p = \sum_{j=m}^{n}  L_{P_j}(0) = \lim_{s \rightarrow 0} \left( \mathbf{A}_{p+1}^{-1}(s)\mathbf{N}^{(0)}_{p+1}\textbf{1}^T \right) \label{eqn32}
\end{eqnarray}
where $\textbf{1}$ is the row array of ones. Here the matrix inverse is the part that is costly and numerically unreliable for large $p$.

Let us consider the most general expression for $\mathbf{A}(s)$ shown above. Thus, we can still use Equation (\ref{eqn32}) to quantify $MTTDL_p$. However, the expressions given in Theorem 3.1 for efficient MTTDL computation holds except $\phi_p(0)$ which did not depend on $\gamma_i$. We provide the following recursive relation without proof for computing $\phi_p(0)$. The proof can be obtained using induction for the known values of $\xi_t$.  For $1 \leq t \leq p$ and the initial condition $\phi_0(0) = n\lambda_0$, we have the following recursion
\begin{align}
\phi_t(0) &= \prod_{i=0}^t \lambda_i(n-i) \nonumber \\
& + (t\mu_{t-1} + \lambda_t(n-t))\Bigg[\phi_{t-1}(0) - \prod_{i=0}^{t-1} \lambda_i(n-i) \nonumber \\ & + \gamma_{t-1}\left(\prod_{i=0}^{t-2}(\gamma_i + \lambda_i(n-i)) + \xi_t \right)\Bigg]  \label{Eqn30}
\end{align}
where $\xi_t \geq 0$ is some small number whose closed form expression is an open problem. However, using symbolic algebra tools, we can obtain few initial evaluations $\xi_1=\xi_2=0$ and $\xi_3 = \gamma_0\mu_0$.

\vspace{1mm}
\begin{corollary}
If $\gamma_j = 0$ for $0 \leq j \leq p-1$, we have $\phi_p(0) = \prod_{i=0}^p \lambda_i(m+i)$.
\end{corollary}
\begin{proof}
  Due to hypothesis $\gamma_j = 0$ for $0 \leq j \leq p-1$, we do not need to worry about $\xi_t$s since they cancel out. Then, it is easy to verify that $\phi_t(0) = \prod_{j=0}^t \lambda_j(n-j)$. By setting $t=p$ and using the change of variables $i=p-j$, the result follows.
\end{proof}

Thus, above corollary verifies Equation (\ref{Phip}). We observe that if we set $\xi_t = 0$ for  $1 \leq t \leq p$, we get $\phi_p^*(0) \leq \phi_p(0)$ where
\begin{align}
\phi_t^*(0) &= \prod_{i=0}^t \lambda_i(n-i) + \\
& (t\mu_{t-1} + \lambda_t(n-t))\Bigg[\phi_{t-1}^*(0) - \prod_{i=0}^{t-1} \lambda_i(n-i) \\ & + \gamma_{t-1}\prod_{i=0}^{t-2}(\gamma_i + \lambda_i(n-i))\Bigg] \label{Eqn31}
\end{align}
with $\phi_0^*(0) = n\lambda_0$. This implies that if we replace $\phi_p(0)$ with $\phi_p^*(0)$ we can find the upper bound for $MTTDL_p$ as follows.
\begin{eqnarray}
MTTDL_p &=& \sum_{j=m}^{m+p}  L_{P_j}(0) \leq \frac{1}{\phi_p^*(0)} \sum_{x=0}^{p} \prod_{\substack{j=0 \\ j \not=x}}^{p}\Lambda_x^{p}(j) \nonumber
\end{eqnarray}
where equality strictly holds for $p=1,2$. For $p>2$, the error value in the overestimation (the upper bound) can be found using the following closed form expression  which quantifies the relationship between $\phi_p^*(0)$ and $\phi_p(0)$.

\vspace{1mm}
\begin{theorem} \label{Thm33}
Let $\phi_p^*(0)$ be an underestimator of  $\phi_p(0)$ as defined above, we have
\begin{eqnarray}
\phi_p(0) - \phi_p^*(0) = \sum_{i=3}^p \gamma_{i-1}\xi_{i} \prod_{j=i}^p (j \mu_{j-1} + \lambda_j(n-j)) \nonumber
\end{eqnarray}
\end{theorem}

\begin{proof}
  It is sufficient to prove the following for $1 \leq t \leq p$,
  \begin{eqnarray}
  \phi_t(0) - \phi_t^*(0) = \sum_{i=1}^t \gamma_{i-1}\xi_{i} \prod_{j=i}^t (j \mu_{j-1} + \lambda_j(n-j)) \label{Eqn33}
\end{eqnarray}

For $t=1$, we have $\phi_1(0) - \phi_1^*(0) = \gamma_0\xi_1(\mu_0+\lambda_1(n-1))$. This is easy to  verify because we have $\phi_0^*(0) = \phi_0(0) = n\lambda_0$. Now suppose that Equation \ref{Eqn33} holds for $t$, and let us show that the same holds for $t+1$. Using Equations (\ref{Eqn30}), (\ref{Eqn31}) and the hypothesis (\ref{Eqn33}) we have $\phi_{t+1}(0) - \phi_{t+1}^*(0)$ equals
\begin{eqnarray}
  &=& ((t+1)\mu_{t} + \lambda_{t+1}(n-t-1)) \left[ \phi_t(0) - \phi_t^*(0) + \gamma_{t}\xi_{t+1}\right] \nonumber \\
 &=& \sum_{i=1}^t \gamma_{i-1}\xi_{i} \prod_{j=i}^{t+1} (j \mu_{j-1} + \lambda_j(n-j)) + \nonumber \\
 && \ \ \ \ \ \ \ \ \ \gamma_{t}\xi_{t+1}((t+1)\mu_{t} + \lambda_{t+1}(n-t-1)) \\
 &=& \sum_{i=3}^{t+1} \gamma_{i-1}\xi_{i} \prod_{j=i}^{t+1} (j \mu_{j-1} + \lambda_j(n-j))
 \end{eqnarray}
 which follows from the fact that $\xi_1 = \xi_2 = 0$. The proof completes if we let  $t=p$.
\end{proof}

\subsection{MTTDL with hard errors}

The general Markov model analyzed earlier is pretty useful for advanced reliability calculations. Here we give one of the simplest improvements over the classical reliability modeling, namely the hard errors that our Markov model can easily incorporate. Hard errors in modern storage arrays are observed to be necessary when the system operates in the critical
mode i.e., a state in which one more device failure leads to total system crash and/or
data loss. This requirement is easily covered by the general Markov model introduced in this study.
Let $\nu$ represent the probability of seeing an uncorrectable error per device
read during say a device rebuilt process. Let UCER denote the uncorrectable error rate of the device
(such as $10^{-15}$, expressed in terms of errors per number of bytes or bits read), $\eta$ is typically given by \cite{Hafner}

\begin{eqnarray}
\eta = 1 - (1 - \textrm{UCER})^{\textrm{device capacity}}
\end{eqnarray}

A transition is needed from state $m+1$ to state $F$ in order to model the rate at which the system
encounters an uncorrectable error while reading and/or rebuilding failed device data. Note that the probability of encountering an uncorrectable
error when reading $m$ devices for rebuild (Note here that we assume conventional MDS codes, which may require many device reads and devices encounter uncorrectable errors independently) is given by  
\begin{eqnarray}
P_{UCER} = 1 - (1 - \eta)^m
\end{eqnarray}

Based on the analysis given in \cite{Hafner}, the uncorrectable error rate for an EPG is computed as the product of the rate that a disk fails
when $m+1$ devices are available and $P_{UCER}$. In order to integrate this probability into the general Markov model introduced earlier, we have to  make the following replacements
\begin{eqnarray}
(m+1)\lambda_{p-1} &\Rightarrow& (m+1)\lambda_{p-1}(1-P_{UCER}) \nonumber \\ && =  (m+1)\lambda_{p-1}(1 - \eta)^m \\
\gamma_{p-1} &\Rightarrow& (m+1)\lambda_{p-1}P_{UCER} \nonumber \\ && =  (m+1)\lambda_{p-1}(1-(1 - \eta)^m)
\end{eqnarray}
and $\gamma_0 = \gamma_1 = \dots = \gamma_{p-2} = 0$. For this special case, the error expression ($\phi_p(0) - \phi_p^*(0)$) in Theorem \ref{Thm33} can be reduced to
\begin{eqnarray}
(m+1)\lambda_{p-1}(1-(1 - \eta)^m)\xi_{p}(p\mu_{p-1} + m \lambda_p) \\
\approx m(m+1)\eta\lambda_{p-1}\xi_{p}(p\mu_{p-1} + m \lambda_p)
\end{eqnarray}
which implies that the error in our efficient calculation of $MTTDL_p$ can be controlled by the uncorrectable error rate of the device.

 \begin{figure*}[t!]
\centering
\includegraphics[height = 60mm, width=\columnwidth]{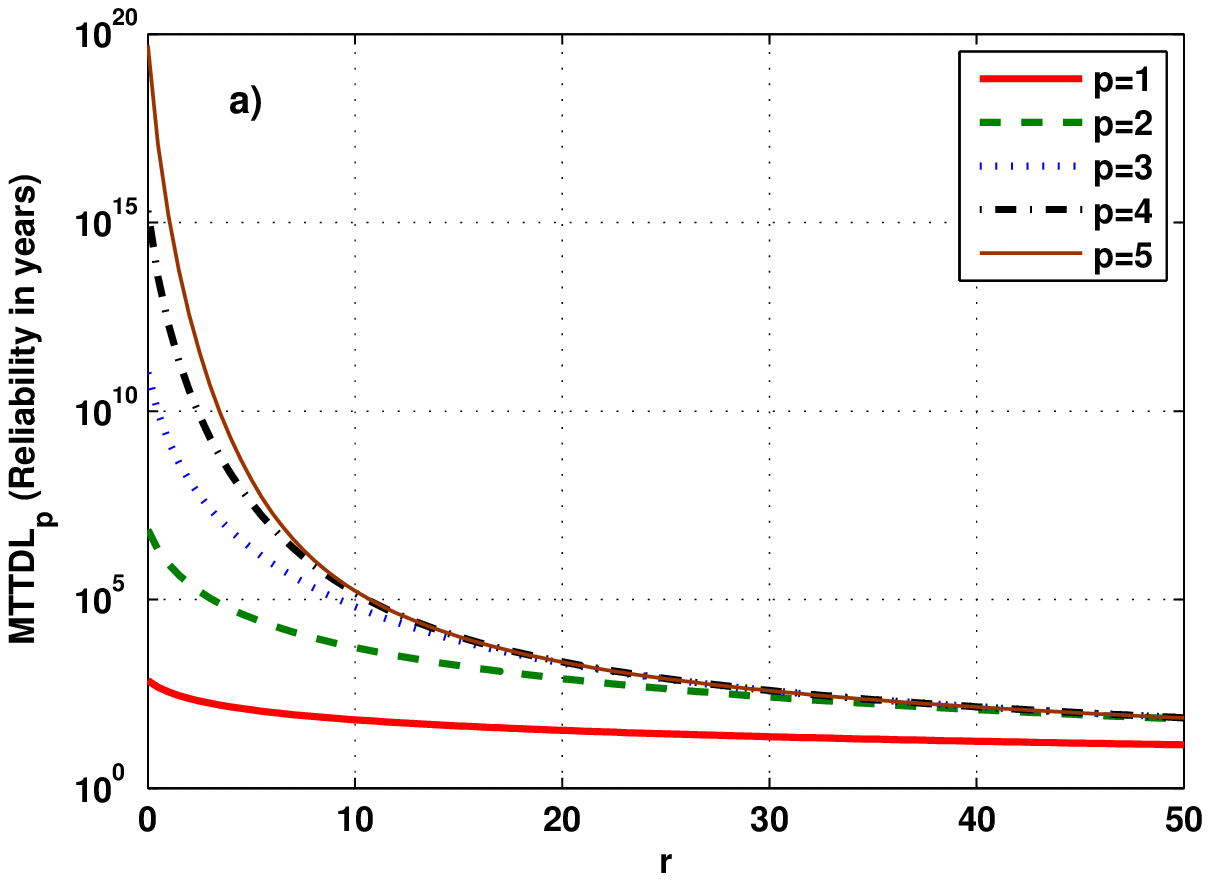}
\includegraphics[height = 60mm, width=\columnwidth]{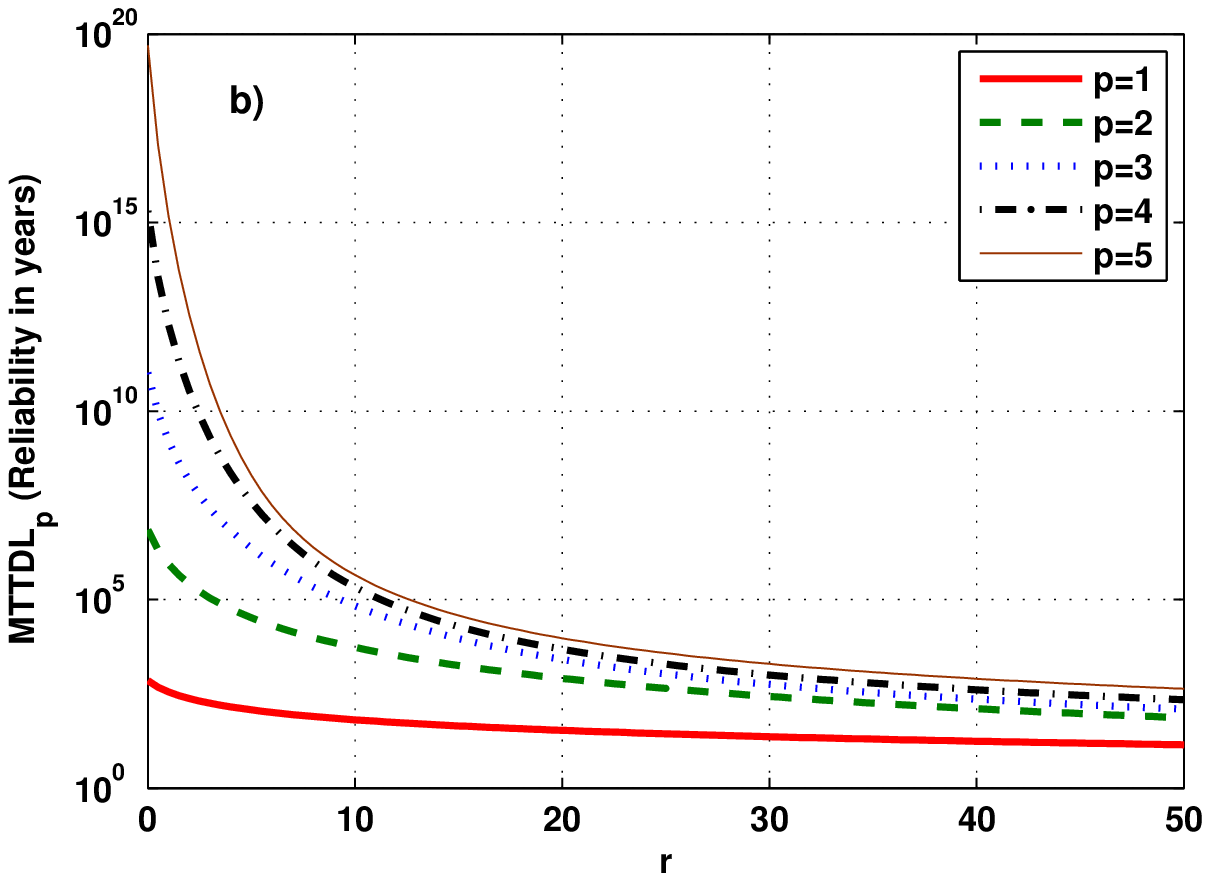}
\caption{Reliability performance results assuming dependent failure rates with exponential failure growth ($\lambda_{max} \rightarrow \infty$, shown in \textbf{a)} ) and logistic failure growth ($\lambda_{max} = 10^{-1}$, shown in \textbf{b)} ) models. We use fixed rate repairs for different number of parities and assumed a fixed data block of size $m=200$ disks in an EPG.}\label{fig:PLOT1}
\end{figure*}

\subsection{A recursive relation for $MTTDL_p$}

For a disk array having a fixed size of $n$ disks, $p$ of which store the parity information, it might be of interest to derive a recursive relationship for $MTTDL_p$. Such a relationship might be useful to predict the additional performance gain through adding extra parity disk into the system. To simplify our analysis, let us assume $\gamma_j = 0$ for $0 \leq j \leq p-1$.

\vspace{1mm}
\begin{theorem} \label{thmB}
For a disk array having a fixed size of $n$ disks, $p$ of which store the parity information, $MTTDL_p$ satisfies the following recursive relationship.
\begin{eqnarray}
MTTDL_{p+1} &=& MTTDL_p + \frac{(p+1)\mu_p}{\lambda_{p+1}(m-1)} MTTDL_p \nonumber \\ && + \frac{1}{\lambda_{p+1}(m-1)}.
\end{eqnarray}
\end{theorem}

\begin{proof}
  The proof is given in Appendix B.
\end{proof}

 As can be seen, this performance improvement is a function of the failure rate $\lambda_{p+1}$ and the repair rate $\mu_p$. Usually, repair rates hardly vary whereas, the failure rates increase as more disks fail in the system.  As long as the storage system satisfies $\lambda_{p+1}(m-1) \gg \max\{\mu_p(p+1),1\}$ \st{then}, we have $MTTDL_{p+1} \rightarrow MTTDL_{p}$ i.e., adding extra parity disk does not improve the reliability of the whole disk array system as the number of parity blocks tends to large numbers (e.g. exponential failure growth). On the other hand, this relationship is not necessarily satisfied in many storage settings (and failure growth models), and yet we will demonstrate that adding parity blocks shall only slightly improve the reliability of the system using more realistic failure growth models (e.g. logistic failure growth).  These arguments will be numerically supported for few failure rate growth models in subsection $F$.

 \subsection{MTTDL with initial defective disks}

In previous section, we assumed $P_n(0) = 1, P_j(0) = 0 \ \textrm{for} \ j = m+p-1,m+p-2,\dots,m, \emph{F}$ i.e., all constituent disks are operational at the start of operation. However, it is possible that when we turn the system on, some of the defective disks will not be able to operate as expected (due to infant mortality period). Similar type of behaviour can be observed in the cluster level as well \cite{Cidon}.  Thus in general, for $j = n,n-1,\dots,m, \emph{F}$, we have $P_j(0) = \epsilon_j$ where $0 \leq \epsilon_j \leq 1$.

Since the erasure code is MDS, it does not matter which disk or disks (or nodes in the cluster) were non-operational at the onset. All it matters is the number of operational disks. Suppose that we have $m$ data, $p$ parity disks with $l \leq p$ non-operational disks at the beginning of the operation. This is no different than turning the system on with $m$ data and $p-l$ parity disks all operational at the beginning. Using such an approach, we can compute $MTTDL_{p-l}$ for $l=0,1,\dots,p$. Let $\boldsymbol{\epsilon} = [\epsilon_{m+p} \ \ \dots \ \ \epsilon_{m} \ \ \epsilon_{F} ]$ be the initial probabilities of being in each state, then we have
\begin{eqnarray}
MTTDL_{p,\boldsymbol{\epsilon}} &=& \lim_{s \rightarrow 0} \sum_{l=0}^p    \textbf{A}^{-1}_{p+1}(s) \epsilon_{m+p-l} \textbf{N}^{(l)}_{p+1} \textbf{1}^T \\
&=& \sum_{l=0}^p \epsilon_{m+p-l}  MTTDL_{p-l} \\ &=& \sum_{l=0}^p \epsilon_{m+l}  MTTDL_{l} \label{DefectiveEqn}
\end{eqnarray}
where $MTTDL_0 = 1/m\lambda_0$ and $\epsilon_{F} + \sum_{l=0}^p \epsilon_{n-l} = 1$. Note that we slightly abused the notation and used the following equation for
convenience
\begin{eqnarray}
MTTDL_{p} = MTTDL_{p,[1 \ 0 \ \dots \ 0]_{1 \times p}}
\end{eqnarray}

\begin{figure*}[t!]
\centering
\includegraphics[height = 60mm, width=\columnwidth]{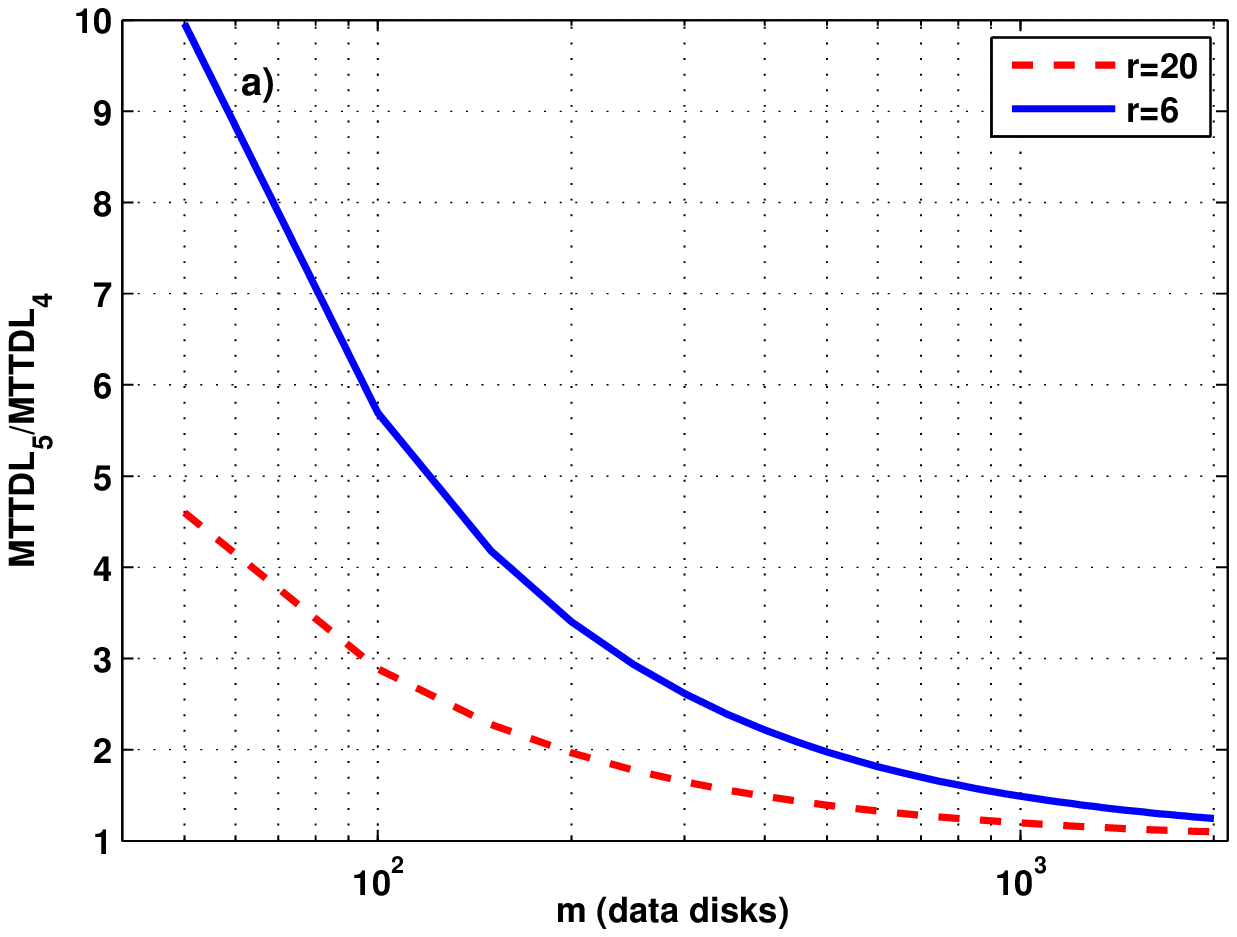}
\includegraphics[height = 60mm, width=\columnwidth]{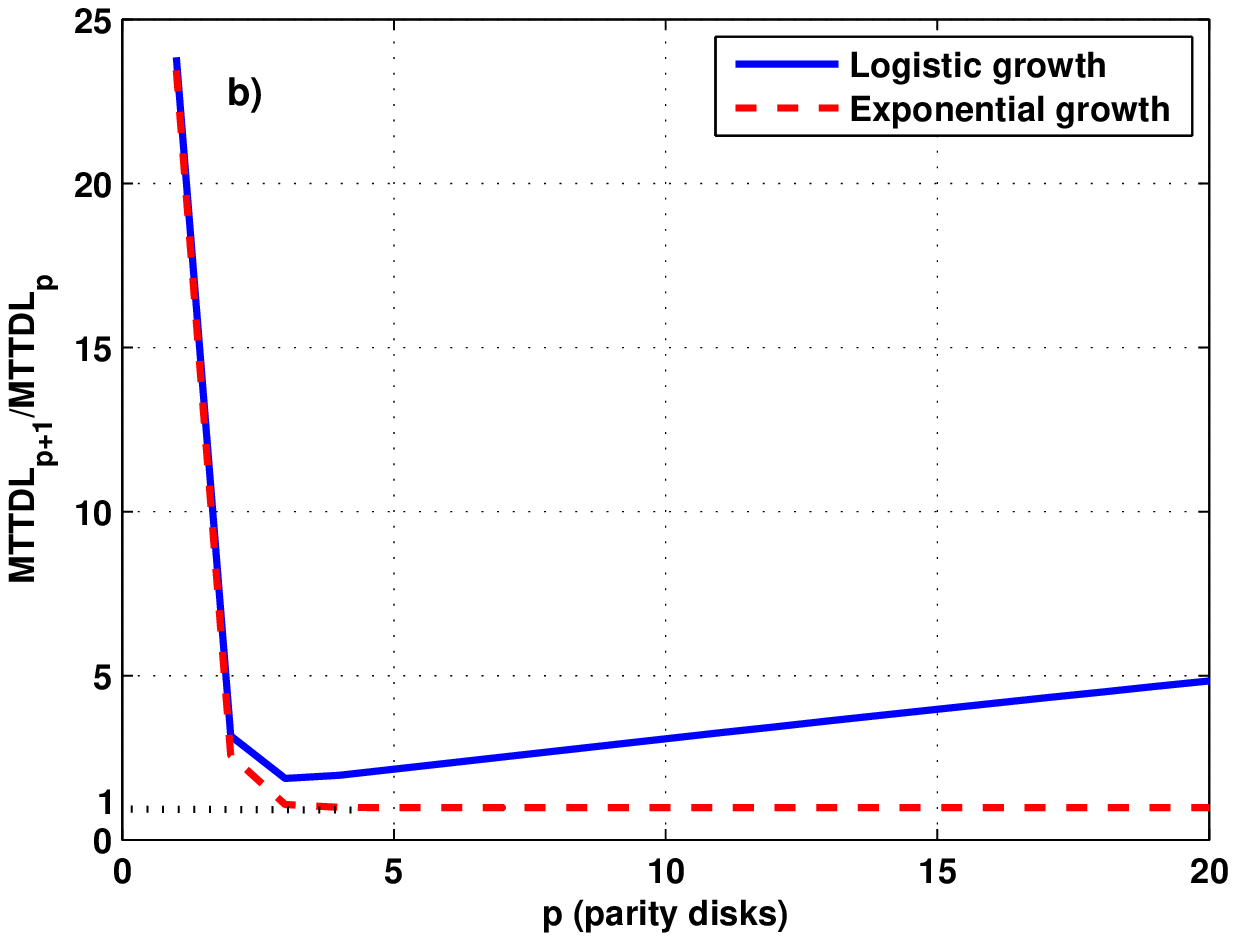}
\caption{ \textbf{a)} $MTTDL_5/MTTDL_4$ ratio, the effect of adding one more parity assuming logistic failure growth ($\lambda_{max} = 10^{-1}$) and using fixed rate repairs. We assumed variable data block sizes. \textbf{b)} $MTTDL_{p+1}/MTTDL_p$ ratio, the effect of adding one more parity assuming exponential failure growth ($\lambda_{max} \rightarrow \infty$) and logistic failure growth ($\lambda_{max} = 10^{-1}$) models as a function of number of parities $p$. We set $m=200$ disks. }\label{fig:PLOT11}
\end{figure*}

\subsection{Real life  failure growth and repair rates}

In this study, we assume disks are in their useful life phase \cite{Modarres}.  A general trend would be to use increased failure rates as we have more failed component disks within the same EPG. Some of the real life observations demonstrate that after a disk failure, the probability of having another failure grows exponentially \cite{Lin}, \cite{GIBSON2}. This suggests that it is reasonable to assume an exponential growth  in the rate of failures after the last failure event. However, after a particular number of failures happen (as we deplete the number of resource/disks), we would expect this growth to stabilize to a constant before the wear-out phase is entered. Such a growth phenomenon is known as \emph{logistic growth} of failure rates \cite{Balakrishnan}. We express the logistic growth for $i=0,\dots,p$ with the following function,
\begin{eqnarray}
\lambda_{i} = \frac{\lambda_0 e^{ir^* }}{1 + (e^{ir^*} - 1) \frac{\lambda_0}{\lambda_{max}}}
\end{eqnarray}
where $\lambda_{max}$ is the maximum failures per hour, i.e., maximum failure rate at which a disk might be failing. If there was no limit on the rate of growth i.e.,  $\lambda_{max} \rightarrow \infty$, we will have the exponential growth (also known as Malthusian growth \cite{Malthus}) expressed as $\lambda_{i} = \lambda_0 e^{ir^*} $ implying the recursive relationship $\lambda_{i+1} = \lambda_{i} e^{r^*} $ for some fixed $r^*>0$. Conventionally, exponential growth is defined with the recursive relationship $\lambda_{i+1} = \lambda_{i} (1 + r) $ for some fixed $r>0$. Therefore, for notation convenience and better visualization we present our results in terms of $r$ using the transformation $\ln(1+r) = r^*$. We also note that it is reasonable to assume that the disk systems are subject to periodic maintenance and therefore we assume a fixed repair rate $\mu = \mu_0 = \dots = \mu_{p-1}$ i.e., $\mu$ repairs per hour.
In general, the repair rate is a function of the erasure code construction, the period with which the system checks for failures, the time it takes to transfer necessary information to recompute the failed disk data and re-balance the distributed storage system. Thus, the repair operation can be expedited by selecting appropriate erasure codes and system maintenance parameters.


\begin{table*}[h]
\small
\centering
\caption{A generic (18,12) MDS code and few Pyramid codes with the associated recoverability/efficieny characteristics (\cite{Huang}) \label{table01}.}{
\begin{tabular}{|l|l|l|l|l|l|l|l|l|}
\hline
\multicolumn{2}{|c|}{Number of failed symbols/blocks}                 & 0   & 1    & 2    & 3    & 4     & 5     & 6     \\ \hline
\multirow{2}{*}{Generic MDS Code}               & Recoverability (\%) & 100 & 100  & 100  & 100  & 100   & 100   & 100   \\ \cline{2-9}
                                                & Avg. read overhead  & 1.0 & 1.61 & 2.22 & 2.83 & 3.44  & 4.06  & 4.67  \\ \hline
\multirow{2}{*}{Pyramid Code (PC)}       & Recoverability (\%) & 100 & 100  & 100  & 100  & 100   & 94.12 & 59.32 \\ \cline{2-9}
                                                & Avg. read overhead  & 1.0 & 1.28 & 1.56 & 1.99 & 2.59  & 3.29  & 3.83  \\ \hline
\multirow{2}{*}{Generalized PC (GPC)} & Recoverability (\%) & 100 & 100  & 100  & 100  & 100   & 94.19 & 76.44 \\ \cline{2-9}
                                                & Avg. read overhead  & 1.0 & 1.28 & 1.56 & 1.99 & 2.59  & 3.29  & 4.12  \\ \hline
\multirow{2}{*}{GPC w/o global symbols}         & Recoverability (\%) & 100 & 100  & 100  & 100  & 97.94 & 88.57 & 65.63 \\ \cline{2-9}
                                                & Avg. read overhead  & 1.0 & 1.28 & 1.56 & 1.87 & 2.32  & 2.93  & 3.85  \\ \hline
\end{tabular}}
\end{table*}
\normalsize

Using the previous assumptions and our closed form expressions derived earlier, let us provide few results for $MTTDL_p$ for $p=1,2,3,4$ and $5$ using different $r$ values and growth rates. Let us assume we have $\lambda_0 = 4 \times 10^{-6}$ failures and $\mu=4$ disk repairs per hour in an EPG that contains $m=200$ disks for raw data storage. All disks are assumed to be operational at the beginning.

Fig. \ref{fig:PLOT1}.a demonstrates that with increasing parity, we dramatically increase the reliability values if the failure rates do not change as we have more and more disk failures in the EPG i.e., independent failure rates. However in case of dependent failures using exponential failure growth model with increasing $r$, the MDS parity schemes become quickly obsolete in that adding more parity is nothing but a waste of resources. For example, adding the fifth parity disk into the disk array which is already protected by four parity disks do not provide any improvement in terms of average reliability statistics when $r=20$. On the other hand, Fig. \ref{fig:PLOT1}.b shows that if we have a logistic failure growth for the constituent disks ($\lambda_{max} = 10^{-1}, r=20$), adding parity helps improve the system performance but this performance improvement is not substantial as predicted by independent failure models \cite{Burkhard}.

In Fig. \ref{fig:PLOT11}.a, we show how much we gain by adding one more parity to an EPG which already has four parity disks for failure protection. We vary the EPG size to see the effect of the EPG size on the reliability performance for two different values of $r$. We observe that as $m$ gets larger, adding an extra parity become almost useless for both values of $r$.  In Fig. \ref{fig:PLOT11}.b, we plot the relative reliability gain of an EPG ($m$ = 200 disks) protected by $p$ parity disks, by adding one more parity disk to the array. We assumed both exponential and logistic failure growths and observed that with exponential growth, there is a limit to the number of parity disks that will be useful in terms of MTTDL performance. After adding four parity disks, we reach the maximum number of parity disks that can benefit the disk array in terms of failure protection. The story changes slightly if we assume logistic failure growth model. As can be seen, adding more parity disks beyond four only slightly helps the reliability performance of the disk array. Although it is not shown explicitly, in order to get the same performance gain we obtain by going from single parity to double parity protection using logistic growth model, we need to have almost 120 additional parity disks. This demonstrates an instance of a very inefficient and possibly very complex protection scheme since we only have $m=200$ data disks and in order to get some real gains by adding parity disks beyond two, we need almost 120 parity disks.

\subsection{Average Read Overhead and Repair Rates}

Systematic erasure codes include the original data blocks as part
of the coded blocks. Thus, accessing any data block can be directly served by the storage
system without further computation. However, if the data block is unavailable, the read operation
has to access a subset of the remaining blocks to recover/compute the missing data block.

The metric \textit{average read overhead}, denoted as $\Phi_j(n)$, represents the average number of extra whole device readings as
an overhead in order to access any unavailable data block (degraded reads) when there are $j$ block failures with a code block length of $n$. Let us consider an example of one block failure ($j=1$) in the $(18,12)$ MDS code to illustrate how this metric is computed. If the failure is a redundant/parity
block ($6/18$ chance), then the data blocks can be accessed directly, so the average
read overhead is 1. Otherwise, the failure shall be a data block ($12/18$ chance) and the read
overhead is twelve for the failed data block and one for the rest of the eleven data blocks.
Hence, the average read overhead is $(12+11)/12$. Altogether, the average read overhead is given by
$\Phi_1 = 1 \times 6/18 + (12+11)/12 \times 12/18 \approx 1.61$.

The following theorem generalizes the average read overhead for any $(n,m)$ systematic block code with average access pattern $\{ S_k^{(n)} \}$ where $k$-th data block can be computed by accessing at least a subset $S_k^{(n)}$ of available blocks for $1 \leq k \leq m$.

\vspace{1mm}
\begin{theorem} \label{thm35}
The average read overhead for a generic $(n,m)$ systematic block code with fixed rate $m/n$ and access pattern $\{S_k^{(n)}\}$ when we have $j \leq n-m$ failures is given by the following generalized expression
\begin{eqnarray}
\Phi_j(n) = \sum_{i=0}^j \frac{(i\overline{S}(n)+m-i)\binom{m}{i}\binom{n-m}{j-i}}{m\binom{n}{j}} \label{eqn50}
\end{eqnarray}
where $\overline{S}(n) = \sum_k |S_k^{(n)}|$ is the average number of block accesses. Furthermore, $\Phi_j(n)$ can be simplified as follows for fixed rate $m/n$ generic $(n,m)$ block code as $n \rightarrow \infty$.
\begin{eqnarray}
\Phi_j(n) \rightarrow 1 + \frac{(\overline{S}(n)-1)j}{n}
\end{eqnarray}
\end{theorem}

\begin{proof}
The proof is given in Appendix C.
\end{proof}

Using the result of Theorem \ref{thm35}, one can deduce that if the average number of block accesses is constant with growing $n$ \cite{Gummadi}, then the average access overhead will approach to the optimal value of 1 irrespective of how many failures there are. On the other extreme for MDS codes we have $\overline{S}(n) = m$ for all $n$. Thus, $\Phi_j(n) \rightarrow 1 + mj/n$ which clearly shows the relationship between the average read overhead and the rate of the code.

Although the average read overhead is not the only metric affecting the repair process, it is usually the dominant one. Since the more data to access and read for the repair, the more time it takes to repair, it is reasonable to assume repair rates to be inversely proportional to this metric. Inspired from the logistic growth, let us define the repair rate with respect to an MDS code
\begin{eqnarray}
\mu_j &\triangleq& \delta\mu\frac{\log((j+1)\Phi_{j+1}^{MDS}(n))}{\log((j+1)\Phi_{j+1}^{Pyd.}(n))} \nonumber \\
&=&  \log_{(j+1)\Phi_{j+1}^{Pyd.}(n)}\Big((j+1)\Phi_{j+1}^{MDS}(n)\Big)^{\delta\mu}
\end{eqnarray}
where $\mu$ is the nominal rate of the repair per device and $\delta$ is a constant used to model the relative bandwidth constraint (with respect to an MDS code, in which it is normalized
to be unity) to reflect on the repair rates based on average read overhead metric. We denote the average read overhead for MDS code as $\Phi_{j+1}^{MDS}(n)$ and for Pyramid codes as $\Phi_{j+1}^{Pyd.}(n)$. Clearly, this formulation
assumes an inverse exponential relationship between the nominal repair rate and the
average read overhead with respect to an MDS code.

\subsection{A case study: Pyramid Codes for Storage}

An interesting case study would be to apply the generalized Markov model  to one of the modern erasure codes such as Pyramid Codes (PC) of Microsoft Azure Storage \cite{Huang}. Pyramid codes are designed to improve the recovery performance for small scale
device failures and have been implemented in archival storage \cite{Wildani}. Pyramid codes are not MDS codes but are
constructed from standard MDS codes by creating newer parity symbols from subsets of
existing data blocks in order to trade-off the recoverability with coding overhead and the average read
overhead, which are important parameters to optimize for a storage application.

In principle, a pyramid code constructs local data sets and generates local parities for these sets based on MDS codes. Additionally, global parities are also generated to span all of the data set for stronger protection against failures. More details about the various construction techniques for pyramid codes can be found in \cite{Huang}. Let us use a ($n=16$, $m=12$) MDS code as the basis for a set of $(18,12)$ pyramid codes given in table \ref{table01}.

We shall either use the computed values of $\Phi_j(n)$ in \cite{Huang} or compute them using Theorem \ref{thm35} (both attains the same values) for our MTTDL evaluations. Using the generalized Markov model of the previous section, let us further assume a homogenous repair strategy is used
in the system. Some results are shown in table \ref{table02} using a nominal
repair rate $\mu = 1/168$ (1 week mean repair time), $\eta = 10^{-3}$ and $\delta = 20$. From the table, we observe that basic
and generalized pyramid codes provide better durability numbers thanks to their efficient repair mechanisms. For a given space efficiency, this is achieved  due to improved access efficiency by sacrificing the recoverability. It has been shown that sacrificing small recoverability can help the system to exploit a huge advantage in average read overhead and hence the MTTDL. On the other hand, as $\lambda$ gets close to zero ($\forall i, \lambda_i = \lambda$), the frequency of repairs go down and hence the advantage of pyramid codes with global symbols diminishes. This can be observed with the MTTDL results given
for $\lambda = 1/1200000$. However, note that basic and generalized PCs still outperform GPC without global symbols, demonstrating the fact that global symbols are quite crucial in pyramid codes for maintaining a desired level of durability. It is also important to notice that the recoverability performance of GPC without global symbols is adversely affected in a dramatic fashion (particularly some 4--failure combinations could not be recovered, see Table \ref{table01}) which leads to degradation in MTTDL performance.

\begin{table}[t!]
\footnotesize
\centering
\caption{MTTDL (hours) for basic MDS as well as various Pyramid codes \cite{Huang} \label{table02}. $K = 10^3$ and $M = 10^6$.}{
\begin{tabular}{|l|l|l|l|l|}
\hline
\multicolumn{1}{|c|}{Failure/Nominal Repair rates} & $\lambda=\frac{1}{200K}$ & $\lambda=\frac{1}{500K}$ & $\lambda=\frac{1}{1.2M}$ \\ \hline
MDS Code            & 2.2e+15        & 6.4e+17        & 1.3e+20        \\ \hline
Basic PC (BPC)          & 1.3e+17        & 5.2e+18        & 1.7e+20        \\ \hline
Generalized PC (GPC)   & 1.32e+17       & 5.26e+18       & 1.76e+20       \\ \hline
GPC w/o global syms           & 1.83e+14       & 3e+15          & 4.1e+16        \\ \hline
\end{tabular}}
\end{table}
\normalsize

\section{Disk Arrays, Distributed Storage and Disk Allocation strategies} \label{Distributed}

In real life storage applications, we have many installations of  disk arrays. This ultimately means more failures and possibly more dependency. Therefore, we need efficient coding mechanisms that can help us have decorrelated disk failures to obtain the gains predicted by independent failure models. For this, we will consider a simple distributed storage network example in which disks are allocated to different nodes of the network.

In this study, we compare two allocation strategies of multiple installations of 1-D MDS disk arrays into a given storage network.  These schemes are summarized in Fig. \ref{fig:DIST2}  and Fig. \ref{fig:DIST}. In the former one, the array elements are placed within the same network node whereas the latter design allocates each component disk that belongs to the same EPG to a different network node in an attempt to minimize the correlated failures. In order for a fair comparison, without loss of generality we assume in each allocation policy, nodes contain the same number of disks and hence we assume the number of EPGs is $z = m + p = n$ in the rest of our analysis/simulations. Note that for  $t$-MDS arrays, such an allocation might not be trivial. An example allocation policy for a particular 2-D MDS disk array system is shown in Section 4.3.

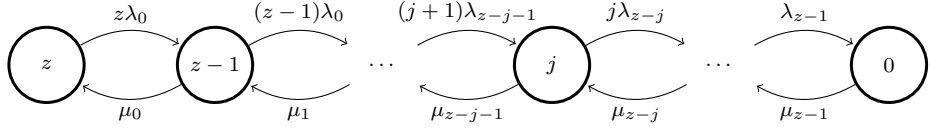
\begin{figure*}[t!]
\centering
\begin{tikzpicture}[node distance=1.2cm]
    \footnotesize
    \node[state, minimum width=1cm, line width=0.4mm]             (n) {$ z $};
    \node[state, minimum width=1cm, line width=0.4mm, right=of n] (n1) {$z-1$};
    \node[state, minimum width=1cm, line width=0.4mm, right=of n1, draw=none] (n2) {$\dots$};
    \node[state, minimum width=1cm, line width=0.4mm, right=of n2] (j) {$j$};
    \node[state, minimum width=1cm, line width=0.4mm, line width=0.4mm, right=of j, draw=none] (n3) {$\dots$};
    \node[state, minimum width=1cm, line width=0.4mm, right=of n3] (m1) {$0$};

    \draw[every loop]
        (n) edge[bend left, auto=left] node {$z\lambda_0$} (n1)
        (n1) edge[bend left, auto=left]  node {$\mu_0$} (n)
        (n1) edge[bend left, auto=left] node {$(z-1)\lambda_0$} (n2)
        (n2) edge[bend left, auto=left] node {$\mu_{1}$} (n1)
        (n2) edge[bend left, auto=left] node {$(j+1)\lambda_{z-j-1}$} (j)
        (j) edge[bend left, auto=left] node {$j\lambda_{z-j}$} (n3)
        (n3) edge[bend left, auto=left] node {$\mu_{z-j}$} (j)
        (j) edge[bend left, auto=left] node {$\mu_{z-j-1}$} (n2)
        (n3) edge[bend left, auto=left] node {$\lambda_{z-1}$} (m1)
        (m1) edge[bend left, auto=left] node {$\mu_{z-1}$} (n3);
\end{tikzpicture}
\caption{A dependent Markov failure model for network storage nodes. In general, we have the relationship $\lambda_0 < \lambda_1 < \dots < \lambda_{z-1}$ to model the increasing failure
rates as we have more and more disk failures within the same column of disks. }\label{fig:MARKOV2}
\end{figure*}

\subsection{Horizontal Allocation of Disks}

Let us assume that we have $z$ installations of the disk array shown in Fig. \ref{fig:RAID}.a. In Fig. \ref{fig:DIST2}, we show how the horizontal allocation of disks is performed in which the array elements are placed within the same network node. Assuming independence between different network nodes\footnote{Since in real life applications, storage network nodes are placed at different physical location and have distinct hardware support. They are also exposed to different environmental conditions with high probability.} and for any EPG having $MTTDL_p^{Hor.}$,  we would like to find the mean time to data loss for the whole storage system.

In a number of experimental observations, the Weibull and gamma distributions are shown to give better approximations to the real lifetime failure characteristics of component disks \cite{Schroeder}.   For the sake of being more realistic, we assume an exponential Time To Failure (TTF) distribution for each component disk and a Weibull TTF distribution for each EPG.  The $i$th EPG  Weibull distribution is given by
\begin{eqnarray}
W_i(t;\omega_i,k_i) &=& \omega_ik_i (\omega_i t)^{k_i-1} e^{-(\omega_it)^{k_i}} \nonumber \\ \mathbb{E}[W_i] &=& \frac{1}{\omega_i} \Gamma \left(1 + \frac{1}{k_i}\right)
\end{eqnarray}
where $k_i>0$ are the shape parameters,  $1/\omega_i > 0$ are the scale parameters of the distribution and $\Gamma(x) = \int_{0}^{\infty} t^{z-1}e^{-t}dt$ is the gamma function.

\begin{figure}[t!]
\centering
\includegraphics[height = 58mm, width=0.8\columnwidth]{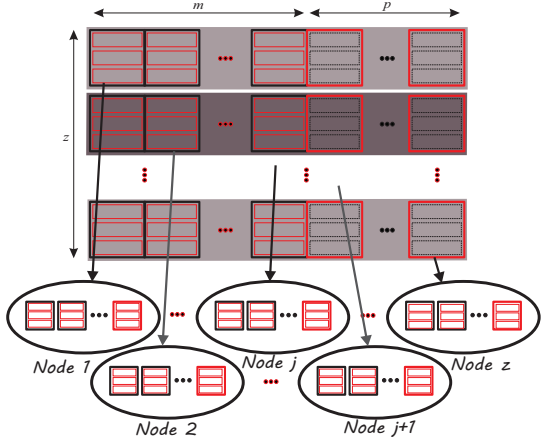}
\caption{Horizontal Allocation: A large number of installations of disk arrays are used commercially to store data. A trivial allocation of EPG arrays into the network nodes is shown. }\label{fig:DIST2}
\end{figure}

Since any EPG failure will result in the whole system failure,  and if we let $W_s$ be the random variable that describes the TTF for the whole system, we have $W_s = \min\{W_1,\dots,W_z\}$. In other words, if any EPG fails, there is no way for the whole storage system to recover the failure. We have for any $t \geq 0$,
\begin{eqnarray}
Pr(W_s > t) &=& Pr(\min\{W_1,\dots,W_z\} > t) \nonumber \\
&=& Pr(W_i > t, i=1,\dots,z) \nonumber \\
&=& \prod_{i=1}^z Pr(W_i > t)  \nonumber \\
&=& \prod_{i=1}^z e^{-(\omega_i t)^{k_i}} = e^{-\sum_{i=1}^z (\omega_i t)^{k_i}} \label{Lasteqn}
\end{eqnarray}

Assuming all EPGs share the same scale parameter $k$, Eqn. (\ref{Lasteqn}) reduces to
\begin{eqnarray}
Pr(W_s > t) = e^{-t^k\sum_{i=1}^z \omega_i^k} = e^{-(\sqrt[k]{\sum_{i=1}^z \omega_i^k}t)^k}
\end{eqnarray}

This implies that the TTF for the whole system has a Weibull distribution with shape parameter $k$ and scale parameter $1/\omega_s = (\sum_{i=1}^z \omega_i^k)^{-1/k}$. The mean time to data loss is given by $\mathbb{E}[W_s] = (1/_s)\Gamma(1 + 1/k)$. Since we assume independence between different network storage nodes and if all EPGs have the same mean time to data loss $MTTDL_p^{Hor.}$, for a given fixed $k$, we have
\begin{eqnarray}
\omega_i = \omega = \frac{\Gamma(1 + 1/k)}{MTTDL_p^{Hor.}} \label{Eqnlamda}
\end{eqnarray}

Using Eqn. (\ref{Eqnlamda}), we deduce
\begin{eqnarray}
\mathbb{E}[W_s] = \frac{\Gamma(1 + 1/k)}{\omega_s}  = \frac{\Gamma(1 + 1/k)}{\omega z^{1/k}} =  \frac{MTTDL_p^{Hor.}}{z^{1/k}}
\end{eqnarray}
which is in accordance with the previously predicted results with $k=1$, i.e., using exponential EPG TTF distributions \cite{RAID}.

\subsection{Vertical Allocation of Disks}

In Fig. \ref{fig:DIST}, we show how the vertical allocation of $z$ (we let $z=n$ for simplicity, it can be generalized to $z>n$) installations of disk arrays are deployed in which each constituent disk that belongs to the same EPG is placed in a different network node. An allocation policy is adapted such that the following criterion is satisfied.

\vspace{1mm}
\begin{definition}
  We define \emph{criterion} to be the case in which a total of $z$ disks are allocated to the different nodes of the storage network consisting of $z$ nodes such that each network node contains only one disk belonging to a particular EPG.
\end{definition}

In Fig. \ref{fig:DIST}, a trivial allocation of disks is considered so that the criterion is satisfied. Other allocations are possible. There are $z$  EPGs and we assume disks in the same storage network node are subject to dependent failure rates because of the same hardware support and environmental conditions etc. On the other hand, we have the storage nodes operating fairly independently.

\begin{figure}[t!]
\centering
\includegraphics[height = 60mm, width=0.7\columnwidth]{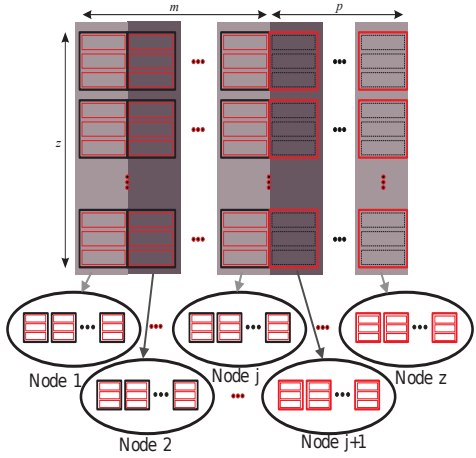}
\caption{\emph{Vertical allocation}: A large number of installations of disk arrays are used to store data. The figure is showing the allocation of component disks of EPGs by putting the first vertical column to the first network node, second column of disks to the second network node and so on.  Other allocations are possible as long as the \emph{\textbf{criterion}} is satisfied.}\label{fig:DIST}
\end{figure}

\begin{table}
\begin{center}
\caption{Disk failure events from  an EPG perspective and associated probabilities\label{Table1}}{
\begin{tabular}{*{4}{c}r}
\hline
\begin{tabular}{@{}c@{}}Number of disk \\ failures in a column ($i$)\end{tabular}  &
\begin{tabular}{@{}c@{}}$\theta_i$= Probability of disk \\ having failure rate $\lambda_i$ \end{tabular}   &
Probability of disk failure  \\
\hline
0 & $z\pi_z/z$ & 0 \\
1   & $(z-1)\pi_{z-1}/z$ & $\pi_{z-1}/z$  \\
2   & $(z-2)\pi_{z-2}/z$ & $2\pi_{z-2}/z$  \\
3   & $(z-3)\pi_{z-3}/z$ & $3\pi_{z-3}/z$  \\
\vdots    & \vdots & \vdots   \\
$z-1$    & $\pi_{1}/z$ & $(z-1)\pi_{1}/z$   \\
$z$    & 0 &  $z\pi_{0}/z$   \\
\hline
\end{tabular}}
\end{center}
\end{table}

Since there is no parity protection across the EPGs (vertical direction), we assume a failure model across the disks that is shown in Fig. \ref{fig:MARKOV2}. The failure model shown in Fig. \ref{fig:MARKOV2} is a truncated generalized birth--death process in which each state label designate the number of operational disks in the network node. For large $z$, the steady state probabilities ($\pi_i$ is the steady state probability of being in state $i$) of such a process is approximated well for $j=0,1,\dots,z$ by \cite{Wang}
\begin{eqnarray}
\pi_{z-j} \approx \frac{\prod_{m=0}^{j-1} \frac{(z-m)\lambda_m}{\mu_m}}{1 + \sum_{k=1}^{z} \prod_{s=0}^{k-1} \frac{(z-s)\lambda_s}{\mu_s}} \label{PIexp}
\end{eqnarray}

For any vertical column of disks (node), if we have all $z$ disks operational (state $z$ in Fig. \ref{fig:MARKOV2}), then all disks will have a failure rate of $\lambda_0$. Similarly, if we have $z-1$ disks operational (state $z-1)$, then all disks will have a failure rate of $\lambda_1$ and so on. From an EPG perspective therefore, any disk is subject to one of the failure rates \{$\lambda_i, i=0,1,\dots,z-1,z$\} with steady state probabilities \{$\pi_{z-i}, i=0,1,\dots,z-1,z$\}, respectively. We define $\lambda_z \triangleq 0$ for completeness, i.e., if all disks are already failed, the rate of failure for remaining disks is zero because there remains no disk to fail.

Suppose that the probability of a failure of a disk in an EPG is $\rho_i$ given that we have
$i$ disk failures in any column of disks (node).   Since for each disk  it is equally likely to have the failure, we have  $\rho_i = i/z$. Similarly, the probability a disk not failing given  that we have
$i$ disk failures in any column of disks of Fig. \ref{fig:DIST} is $1-\rho_i$. Since the probability of having $i$ disk failures is $\pi_{z-i}$, the unconditional probabilities will be given by Table \ref{Table1} where we list all the possibilities. The total disk failure probability $\rho_F$  is then given by averaging over the number of disk failures $i$,
\begin{eqnarray}
\theta_F = \sum_{i=0}^z \rho_i Pr\{\textrm{$i$ disk failures}\} = \sum_{i=0}^z \frac{i\pi_{z-i}}{z}
\end{eqnarray}

Thus, we have the probability distribution given by the probabilities $\{\theta_0, \theta_1, \dots, \theta_z, \theta_F\}$ satisfying $\sum_{i} \theta_i + \theta_F = 1$.
Before the EPG decoding attempts recovery, suppose that we have $\nu$ of such disk failures (conditioning on $\nu$, assuming all different $\nu$ combinations are equally likely). The probability of having  $\nu$ of such failures is binomially distributed (due to node independence) and given by
\begin{eqnarray}
\epsilon_{n-\nu} = \binom{n}{\nu} \rho^\nu (1-\rho)^{n-\nu}
\end{eqnarray}

Since we are left with $n-\nu$ operational disks, each failing with one of the failure rates $\{\lambda_0, \lambda_1, \dots, \lambda_{z-1}\}$ with probabilities $\{\pi_z, (z-1)\pi_{z-1}/z, \dots, \pi_1/z \}$ i.e., $Pr\{$A disk fails with rate $\lambda_j\} = \theta_j = (z-j)\pi_{z-j}/z$. Thus, the disks in an EPG will fail independently with the average failure rate\footnote{This is a generalization of subdividing a Poisson process. Suppose each arrival in a Poisson process  $\{N(t), t \geq 0\}$ of rate $\lambda$, is sent into one of the
two arrival processes $\{N_1(t), t \geq 0\}$  and $\{N_2(t), t \geq 0\}$ with probabilities $p$ and  $1-p$, respectively. The resulting processes are Poisson with rates $\lambda_1 = p\lambda$ and $\lambda_2 = (1-p)\lambda$, respectively. Thus, the average rate will give us the rate of the original undivided Poisson process, i.e., $\lambda_{avg} = \lambda_1 + \lambda_2 = \lambda$. It is straightforward to extend this argument to generalized subdivisions of the Poisson process.}
\begin{eqnarray}
\lambda_{avg} = \sum_{j=0}^{z} \lambda_j \theta_j + 0.\theta_F = \sum_{j=0}^{z-1} \lambda_j (z-j)\pi_{z-j}/z \label{AvgLamda}
\end{eqnarray}


If we let $\overline{MTTDL}_p$ denote the mean time to data loss using $\{\lambda_0 = \lambda_{avg}, \dots, \lambda_{n-m} = \lambda_{avg}\}$ and $\{\mu_0,\dots,\mu_{n-m-1}\}$, we will have the following mean time to data loss expression for any EPG in the storage system. Averaging over $\nu$ using Equation (\ref{DefectiveEqn}) yields
\begin{eqnarray}
MTTDL_p^{Ver.} &=& \sum_{\nu = 0}^{p} \epsilon_{n-\nu}  \overline{MTTDL}_{p-\nu} \\ &=&  \sum_{\nu = 0}^{p} \binom{n}{\nu} \rho^\nu (1-\rho)^{n-\nu}  \overline{MTTDL}_{p-\nu}
\end{eqnarray}

Finally, using the arguments of the previous section, we can obtain the mean time to data loss of the whole storage system using vertical allocations and Weibull distribution for the EPG TTF given by $MTTDL_p^{Ver.}/z^{1/k}$.

\subsection{Numerical Results}

Since the main objective of this subsection is to show the relative reliability analysis of different allocation policies, we only show performances of one dimensional MDS arrays. The same conclusions can be drawn with larger dimensional disk array systems such as protected by product codes (2-D MDS codes) \cite{GIBSON}. In larger dimensional protection groups however, allocation policies might not be so straightforward to meet our objectives to decorrelate disk failure events as much as possible.

\begin{figure}[t!]
\centering
\includegraphics[height = 70mm, width=\columnwidth]{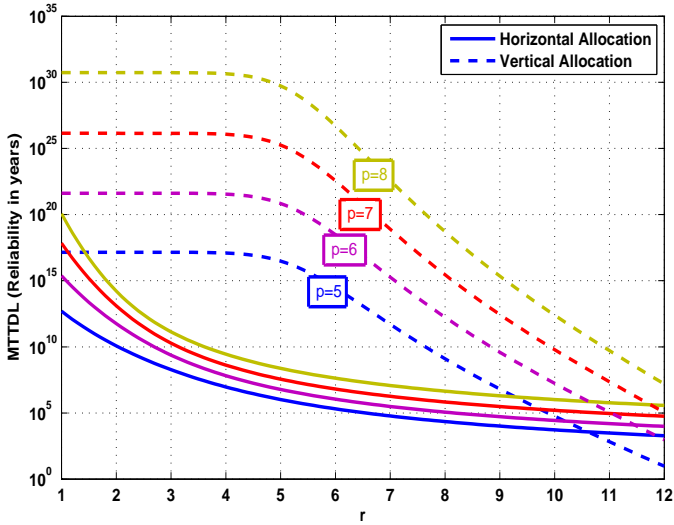}
\caption{Comparison of information dispersal methodologies: Horizontal and Vertical allocations of disks into the network storage nodes.}\label{fig:COMPARISON1}
\end{figure}

We assume $n=200$ and $z=n$ EPGs which amounts to 40000 disks.
The reliability of such a storage system will be presented using two different allocation policies in terms of MTTDL. As previously assumed, let us use disks with $\lambda_0 = 4 \times 10^{-6}$ and $\mu_0 = 4$ hours and a logistic growth with $\lambda_{max} = 3 \times 10^{-2}$. The shape parameter $k = 0.9$ is assumed for TTF Weibull distribution.  In Fig. \ref{fig:COMPARISON1}, we show the results for $p$= 5, 6, 7, 8  using two different allocation policies: \emph{Horizontal} and \emph{Vertical} Allocations. As can be seen there is a limit to the value of $r$ beyond which the failure rates become detrimental and lead the vertical allocation to give performance values below that of horizontal allocation. Of course, this is because there is no error protection in vertical direction. However for $r < 10$, the vertical allocation policy allows greater reliability than the horizontal allocation. In fact, for a range of $r$ i.e., $1 \leq r \leq  5$, MTTDL performance does not show any dramatic change with increasing $r$. This is due to the decorrelation of disk failures for large installation of 1-D MDS disk array systems. Another interesting observation is that using vertical allocation and $p=6$, the performance is almost always greater than that of horizontal allocation and $p=8$ for the $r$ of interest in the same figure. This suggests that by changing the allocation policy, we can save some parity disks and still be able to have the intended reliability for the whole disk array system in a distributed storage setting. Another way of reading the same figure is to look at the effect of adding extra parities to increase the failure tolerance of the system for a given target reliability level. We can observe that for any given reliability target, similar diminishing returns argument applies to the failure tolerance for going from $p$ parities to $p+1$ parities, though it provides different gains for horizontal and vertical allocations.

\section{Conclusion} \label{CNCLSN}

Disk replacement rates show significant correlation between constituent disks that are locally stored in the same storage network node and are found to be no where near manufactures' reported disk failure or replacement rates. Based on the recent survey data and experimental evidences, a more comprehensive and applicable modeling is needed to be able to accommodate for the dependencies in different failure modes and compensate for different data allocations. Additionally, with the advancement of new erasure codes, novel ways of storing data has evolved and has put new metrics in forefront such as data regeneration, read overhead and efficient repair. In this study, we proposed a generalized failure model that can capture realistic parameters and provide more accurate reliability estimations for MDS disk arrays. In particular, we argued that instead of adding more parity locally, it may be more convenient to disperse the data and parity across the network to be able to have disks belonging to the same protection group work relatively independently. Additionally, we have shown how the average read overhead affects the repair rates and hence the reliability of the overall storage system. Although our discussion is simplified by using 1-D MDS codes and straightforward allocation policies, the results can be generalized using the proposed model to larger dimensional and modern coding schemes in which more sophisticated allocation policies might be needed to be able to find ways for decorrelated component disk operations. This is made simple in this study by driving efficient computation of the MTTDL metric.

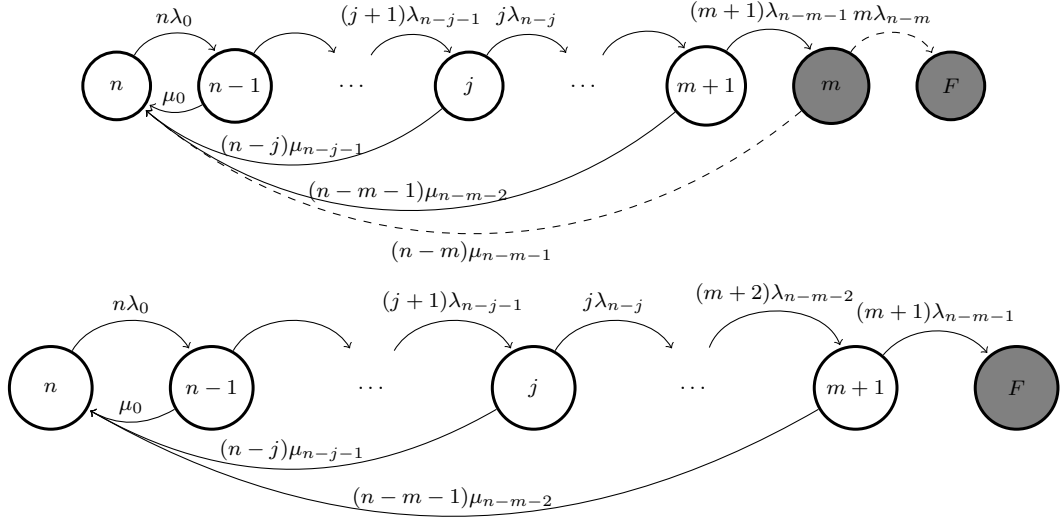
\begin{figure*}[htp!]
\centering
\begin{tikzpicture}[node distance=0.6cm]
    \footnotesize
    \node[state, minimum width=0.9cm, line width=0.4mm]             (n) {$ n $};
    \node[state, minimum width=0.9cm, line width=0.4mm, right=of n] (n1) {$n-1$};
    \node[state, minimum width=0.9cm, line width=0.4mm, right=of n1, draw=none] (n2) {$\dots$};
    \node[state, minimum width=0.9cm, line width=0.4mm, right=of n2] (j) {$j$};
    \node[state, minimum width=0.9cm, line width=0.4mm, line width=0.4mm, right=of j, draw=none] (n3) {$\dots$};
    \node[state, minimum width=0.9cm, line width=0.4mm, right=of n3] (m1) {$m + 1$};
    \node[state, fill=gray, minimum width=1cm, line width=0.4mm, right=of m1] (m) {$m$};
    \node[state, fill=gray, minimum width=0.9cm, line width=0.4mm, right=of m] (f) {$F$};

    \draw[every loop]
        (n) edge[bend left=60, auto=left] node {$n\lambda_0$} (n1)
        (n1) edge[bend left, auto=right]  node {$\mu_0$} (n)
        (n1) edge[bend left=60, auto=left] node {} (n2)
        (n2) edge[bend left=60, auto=left] node {$(j+1)\lambda_{n-j-1}$} (j)
        (j) edge[bend left=60, auto=left] node {$j\lambda_{n-j}$} (n3)
        (n3) edge[bend left=60, auto=left] node {} (m1)
        (m1) edge[bend left=60, auto=left] node {$(m+1)\lambda_{n-m-1}$} (m)
        (m1) edge[bend left=40, auto=right] node {$(n-m-1)\mu_{n-m-2}$} (n)
        (m) edge[bend left=40, dashed, auto=left]  node {$(n-m)\mu_{n-m-1}$} (n)
		(j) edge[bend left=40, auto=right]  node {$(n-j)\mu_{n-j-1}$} (n)
        (m) edge[bend left=60, dashed, auto=left] node {$m\lambda_{n-m}$} (f);

\end{tikzpicture}
\begin{tikzpicture}[node distance=1cm]
    \footnotesize
    \node[state, minimum width=1.1cm, line width=0.4mm]             (n) {$ n $};
    \node[state, minimum width=1.1cm, line width=0.4mm, right=of n] (n1) {$n-1$};
    \node[state, minimum width=1.1cm, line width=0.4mm, right=of n1, draw=none] (n2) {$\dots$};
    \node[state, minimum width=1.1cm, line width=0.4mm, right=of n2] (j) {$j$};
    \node[state, minimum width=1.1cm, line width=0.4mm, line width=0.4mm, right=of j, draw=none] (n3) {$\dots$};
    \node[state, minimum width=1.1cm, line width=0.4mm, right=of n3] (m1) {$m + 1$};
    \node[state, fill=gray, minimum width=1.1cm, line width=0.4mm, right=of m1] (m) {$F$};

    \draw[every loop]
        (n) edge[bend left=60, auto=left] node {$n\lambda_0$} (n1)
        (n1) edge[bend left, auto=right]  node {$\mu_0$} (n)
        (n1) edge[bend left=60, auto=left] node {} (n2)
        (n2) edge[bend left=60, auto=left] node {$(j+1)\lambda_{n-j-1}$} (j)
        (j) edge[bend left=60, auto=left] node {$j\lambda_{n-j}$} (n3)
        (n3) edge[bend left=70, auto=left] node {$(m+2)\lambda_{n-m-2}$} (m1)
        (m1) edge[bend left=50, auto=left] node {$(m+1)\lambda_{n-m-1}$} (m)
        (m1) edge[bend left=30, auto=right] node {$(n-m-1)\mu_{n-m-2}$} (n)
		(j) edge[bend left=30, auto=right]  node {$(n-j)\mu_{n-j-1}$} (n);

\end{tikzpicture}
\caption{Markov failure model for $m$ and $m+1$ data disks where the size of the EPG is fixed and equals to $n$. \emph{F}: Failure state if $m$ data disks are being utilized.}\label{fig:MARKOVproof}
\end{figure*}

\appendices
\section{Proof of Theorem \ref{thm31}} \label{app1}

For a given integer $c \leq p$, we let $\lambda_p = 0, \lambda_{p-1} = 0, \dots, \lambda_{c}=0$ and $\mu_{p-1}=0,\dots, \mu_{c-1} = 0$. Therefore using the matrix $A_{p+2}$, we will have
 \[\left[ \begin{array}{cc}
\textbf{A}_{c+1}   & \textbf{0}  \\
\textbf{0} & s\textbf{I}_{p+2-c-1} \end{array} \right] \left[ \begin{array}{c}
L_{p_{n}}(s) \\
L_{p_{n-1}}(s)\\
\vdots \\
L_{p_{n-(c-1)}}(s)\\
L_{p_{F}}(s)\\
\vdots \\
L_{p_{F}}(s) \end{array} \right] = \left[ \begin{array}{c}
1 \\
0\\
\vdots \\
0 \end{array} \right] \]

\[ \Rightarrow \textbf{A}_{c+1}  \left[ \begin{array}{c}
L_{p_{n}}(s) \\
L_{p_{n-1}}(s)\\
\vdots \\
L_{p_{n-(c-1)}}(s)\\
L_{p_{F}}(s)\\
 \end{array} \right] = \left[ \begin{array}{c}
1 \\
0\\
\vdots \\
0 \end{array} \right] \]
and $sL_{p_{F}}(s) = 0$. This argument implies that the expressions given for $A_{p+2}$ (an EPG using $m$ data disks) can be used to find expressions for $\textbf{A}_{c+1}$ with $c \leq p$ (an EPG using $n+1-c$ data disks).

For $m=1$ (single data disk)  is the trivial case, we focus on the inductive step for $m>1$. For this, we start with considering the Fig. \ref{fig:MARKOVproof}. In  \ref{fig:MARKOVproof} a. we show the generalized Markov failure model using $m$ data disks whereas \ref{fig:MARKOVproof} b. shows the same  model and same size EPG having $m+1$ data disks. In order to go from \ref{fig:MARKOVproof} a. to \ref{fig:MARKOVproof} b., it is straightforward to see that we must set $\mu_{n-m-1} = 0$. This implies if we have $m
+1$ data disks for an $n$-disk EPG, we can recover up to $p-1$ disk failures. If we have more than $p-1$ whole disk failures, no repair can put the EPG back in operation.  Also, from transform domain equations (given above)  by setting $\mu_{n-m-1} = 0$ and $s=0$, we obtain
\begin{eqnarray}
 -\lambda_{n-m-1}(m+1) L_{P_{m+1}}(0)   + \lambda_{n-m}mL_{P_{m}}(0) = P_{m}(0) = 0 \nonumber
\end{eqnarray}
which implies that $\lambda_{n-m-1}(m+1) L_{P_{m+1}}(0) = \lambda_{n-m}mL_{P_{m}}(0)$. Furthermore, we have $L_{P_{m+1}}(0) = L_{P_{m}}(0) \Leftrightarrow\lambda_{n-m-1}(m+1) = \lambda_{n-m}m$.

In the inductive step, let us assume the given expressions are true for an EPG using $m$ data disks and show that they are true for an EPG using $m+1$ data disks. First note that
for $x=0,1,\dots,p-2$, we have
\begin{eqnarray}
\prod_{\substack{j=0 \\ j \not=x}}^{p-1}\Lambda_x^{p}(j) = ((p-1)\mu_{p-2} + \lambda_{p-1}(m+1))\prod_{\substack{j=0 \\ j \not=x}}^{p-2}\Lambda_x^{p}(j)
\end{eqnarray}

It is easy to verify that we have $\phi_p(0) = m\lambda_p \phi_{p-1}(0)$ where $\phi_p(0) = \prod_{i=0}^{p}\lambda_i (m + i) $.  if we let $m \rightarrow m+1$ and $p  \rightarrow p-1$, we have $L_{P_{m}}(0) \rightarrow L_{P_{m+1}}(0)$. Finally, we have  $L_{P_{m+1}}(0) = L_{P_{m}}(0) \Rightarrow\lambda_{n-m-1}(m+1) = \lambda_{n-m}m$.
Now, using these results and the inductive assumption with $\mu_{n-m-1}=0$ and $(m+1)\lambda_{p-1} = m\lambda_p$, we finally have
\begin{eqnarray}
L_{P_{m+p-x}}(0) &=& \frac{(p\mu_{p-1} + \lambda_pm)}{\phi_p(0)}\prod_{\substack{j=0 \nonumber \\ j \not=x}}^{p-1}\Lambda_x^{p}(j) \nonumber \\
&=& \frac{1}{\phi_{p-1}(0)} ((p-1)\mu_{p-2} \nonumber \\ 
&& \ \ \ + \ \lambda_{p-1}(m+1))\prod_{\substack{j=0 \nonumber \\ j \not=x}}^{p-2}\Lambda_x^{p}(j) \\
&=& L_{P_{(m+1)+(p-1)-x}}(0)
\end{eqnarray}
where $L_{P_{(m+1)+(p-1)-x}}(0)$ is the expression obtained from the expression given for $L_{P_{m+p-x}}(0)$ by replacing $m$ with $m+1$ and $p$ with $p-1$.   This simply implies that we can use the same expression to obtain reliability values when $p \rightarrow p-1$ while $n = m + p$ is kept fixed. Now we consider if $x=p-1$. From the inductive assumption, we have
\begin{eqnarray}
L_{P_{m}}(0) &=& \frac{1}{\phi_p(0)}\prod_{i=0}^{p-1} \lambda_i (m + 1 + i) \nonumber \\ &=& \frac{1}{m \lambda_p \phi_{p-1}(0)}\prod_{i=0}^{p-1} \lambda_i (m + 1 + i) \nonumber \\
&=& \frac{1}{ \phi_{p-1}(0)}\prod_{i=0}^{p-2} \lambda_i (m + 2 + i) \label{Eqn48}
\end{eqnarray}

Since $L_{P_{m+1}}(0) = L_{P_{m}}(0) \Leftrightarrow\lambda_{p-1}(m+1) = \lambda_{p}m$, we have  the Eqn. (\ref{Eqn48}).  This completes the proof. $\blacksquare$

\section{Proof of Theorem \ref{thmB}} \label{app2}

For a fixed $n$, let us start with a explicit expression for $MTTDL_{p+1}$ (using $m-1$ data disks),
\begin{align}
\frac{((p+1)\mu_{p} + \lambda_{p+1}(m-1))}{\phi_{p+1}(0)} \sum_{x=0}^{p} \prod_{\substack{j=0 \\ j \not=x}}^{p}\Lambda_x^{p+1}(j)   + \frac{1}{\lambda_{p+1}(m-1)}. \textcolor{blue}{\label{EqnAppB0}}
\end{align}

It is easy to verify that we have
\begin{eqnarray}
\phi_{p+1}(0) = \lambda_{p+1}(m-1)\phi_{p}(0).   \label{EqnAppB1}
\end{eqnarray}

Also, we have the following algebraic manipulation to use
\begin{eqnarray}
\sum_{x=0}^{p} \prod_{\substack{j=0 \\ j \not=x}}^{p}\Lambda_x^{p+1}(j) &=& \sum_{x=0}^{p-1} \prod_{\substack{j=0 \\ j \not=x}}^{p}\Lambda_x^{p+1}(j) + \prod_{\substack{j=0 \\ j \not=p}}^{p}\Lambda_x^{p+1}(j) \nonumber \\
&=& (\lambda_p m + p \mu_{p-1}) \sum_{x=0}^{p-1} \prod_{\substack{j=0 \\ j \not=x}}^{p-1}\Lambda_x^{p}(j) \nonumber \\ && +  \prod_{j=0}^{p-1} \lambda_j (m+1+j). \label{EqnAppB2}
\end{eqnarray}

Let us plug Eqns. (\ref{EqnAppB1}) and (\ref{EqnAppB2}) into the Eqn. (\ref{EqnAppB0}), we have
\begin{eqnarray}
MTTDL_{p+1} &=& \frac{(p+1)\mu_{p} + \lambda_{p+1}(m-1)}{\lambda_{p+1}(m-1)} MTTDL_{p} \nonumber \\ && + \frac{1}{\lambda_{p+1}(m-1)}
\end{eqnarray}
as desired.  $\blacksquare$

\section{Proof of Theorem \ref{thm35} } \label{app3}

Let us assume we have $0 \leq j \leq n-m$ failed blocks. Of these failures, let us suppose that $i$ out of $j$ failures happen in the $m$ data blocks and the remaining $j-i$ failures happen in the rest of $n-m$ parity blocks. The overhead conditioned on $j$ failures depends on which $i$ data blocks are failed because $S_k^{(n)}$s might be different. Let us consider each $i$ failure combinations and for each we sum the total number of accesses. This expression is given by
\begin{eqnarray}
\sum_{s=1}^{\binom{m}{i}}\sum_{c \in C_s} |S_c^{(n)}| = \binom{m}{i} \frac{i}{m} \sum_{k=1}^m |S_k^{(n)}| = \binom{m}{i}i\overline{S}(n)
\end{eqnarray}
where $C_s$ is the set of indexes corresponding to $s$-th combination of all $\binom{m}{i}$ combinations and $\overline{S}(n) = \sum_k |S_k^{(n)}|$ is the average number of block accesses. On the other hand, for unfailed $m-i$ blocks, we have access overhead of unity. Since we sum all the different combinations, we finally have $\binom{m}{i}(m-i)$ access read overhead. Since the rest of the $j-i$ failures on the parity blocks can happen in different combinations, we have the following total number of accesses
\begin{eqnarray}
\left( \binom{m}{i}i\overline{S}(n) + \binom{m}{i}(m-i) \right)\binom{n-m}{j-i}.
\end{eqnarray}

This number shall be divided by all possible combinations $\binom{n}{j}$ multiplied by the number of data blocks $m$. Since selecting any data block is equally likely, we have division by $m$. Finally, we sum over all possible $i$ to find the unconditional average overhead which is given by Equation (\ref{eqn50}).

If we inspect the hypergeometric probability term in Equation (\ref{eqn50}), it is easy to see that we have
\begin{align}
\frac{\binom{m}{i}\binom{n-m}{j-i}}{\binom{n}{j}} = \binom{j}{i} \prod_{r=1}^i \frac{m-i+r}{n-i+r} \prod_{s=1}^{j-i} \frac{n-m-(j-i)+s}{n-j+s}
\end{align}
which yields
\begin{eqnarray}
\lim_{n \rightarrow \infty} \frac{\binom{m}{i}\binom{n-m}{j-i}}{\binom{n}{j}} = \binom{j}{i} (m/n)^i (1-m/n)^{j-i}
\end{eqnarray}
where convergence happens as $n \rightarrow \infty$ for constant $m/n$. We now can express Equation (\ref{eqn50}) as follows
\begin{align}
\lim_{n \rightarrow \infty} \Phi_j(n) = \sum_{i=0}^j \left( 1 + \frac{i(\overline{S}(n)-1)}{m}\right)  \binom{j}{i} \left(\frac{m}{n}\right)^i \left(1-\frac{m}{n}\right)^{j-i}
\end{align}
which implies $\lim_{n \rightarrow \infty} \Phi_j(n) = 1 + \frac{(\overline{S}(n)-1)}{m} j (m/n)$ from which the result follows. $\blacksquare$


\ifCLASSOPTIONcaptionsoff
  \newpage
\fi

%


\end{document}